\documentclass[11pt]{article}
\usepackage[margin=1in, centering]{geometry}

\usepackage{graphicx} % Required for inserting images
\usepackage{ktmacros}
\usepackage{amsmath}
\usepackage{esint}
\usepackage{mathtools}
\usepackage{amssymb}
\usepackage{amsthm}
\usepackage{xcolor}
\usepackage{bbm}
\newcommand{\Vol}{\mathrm{Vol}}
\newcommand{\Area}{\mathrm{Area}}
\newcommand{\TruncLap}{\mathrm{TruncLap}}
\newcommand{\COV}{\mathrm{cov}}
\newcommand{\calD}{\mathcal{D}}
\newcommand{\calA}{\mathcal{A}}
\newcommand{\calB}{\mathcal{B}}
\newcommand{\calX}{\mathcal{X}}

\title{Fingerprinting Codes Meet Geometry: Improved Lower Bounds for Private Query Release and Adaptive Data Analysis}
\author{Xin Lyu\footnote{Research done while the author was at Apple. Email: xinlyu@berkeley.edu}\\ UC Berkeley \and Kunal Talwar\footnote{Email:kunal@kunaltalwar.org}\\Apple}
\date{}

\begin{document}

\maketitle

\newcommand{\xin}[1]{{\color{red} Xin: #1}}
\newcommand{\kunal}[1]{{\color{purple} Kunal: #1}}
% \maketitle

\thispagestyle{empty}

\begin{abstract}

Fingerprinting codes are a crucial tool for proving lower bounds in differential privacy. They have been used to prove tight lower bounds for several fundamental questions, especially in the ``low accuracy'' regime. Unlike reconstruction/discrepancy approaches however, they are more suited for query sets that arise naturally from the fingerprinting codes construction. In this work, we propose a general framework for proving fingerprinting type lower bounds, that allows us to tailor the technique to the geometry of the query set.
Our approach allows us to prove several new results, including the following.
\begin{itemize}
  \item We show that any (sample- and population-)accurate algorithm for answering $Q$ arbitrary adaptive counting queries over a universe  $\mathcal{X}$ to accuracy $\alpha$ needs $\Omega(\frac{\sqrt{\log |\mathcal{X}|}\cdot \log Q}{\alpha^3})$ samples, matching known upper bounds. This shows that the approaches based on differential privacy are optimal for this question, and improves significantly on the previously known lower bounds of $\frac{\log Q}{\alpha^2}$ and $\min(\sqrt{Q}, \sqrt{\log |\mathcal{X}|})/\alpha^2$.
  \item We show that any $(\eps,\delta)$-DP algorithm for answering $Q$ counting queries to accuracy $\alpha$ needs $\Omega(\frac{\sqrt{ \log|\mathcal{X}| \log(1/\delta)} \log Q}{\eps \alpha^2})$ samples, matching known upper bounds up to constants. Our framework allows for proving this bound via a direct correlation analysis and improves the prior bound of~\cite{BunUV14} by $\sqrt{\log(1/\delta)}$. % The new bound is tight up to constants. %and qualitatively ($\ell_2$ error vs. $\ell_\infty$ error) on the bound proven by~\cite{BunUV14} using composition.
  \item For privately releasing a set of random $0$-$1$ queries, we show tight sample complexity lower bounds in the high accuracy regime.
\end{itemize}

  In the low accuracy regime, the picture is more complex. For random queries, we show that there is a discontinuity in the sample complexity. For $d$ random queries over a universe $\calX$, the sample complexity grows as $\Theta_{\eps, \delta}(\frac{1}{\alpha^2})$, with no dependence on $d$ or $|\mathcal{X}|$. This new sample complexity bound, based on sparse histograms, is asymptotically better than known lower bounds for CDP. However, at $\alpha \approx \sqrt{\log |\mathcal{X}|}/\sqrt{d}$, the sample complexity jumps to $\Theta_{\eps,\delta}(\sqrt{d}/\alpha)$.

\end{abstract}

% Maybe downplay the techniques advances. Emphasize the questions.
% In techniques, discuss the exponential family idea from Kamath et al. (but limited to hypercube), the Stoke's theorem idea from Harvey et al. (but limited to covariance estimation), both for population questions.

% Understand better Assouad and Kramer-Rao.

% Intro: expand on the abstract. State list of questions.

% Then put our contributions sections.

% Then techniques section.
% Need to discuss Score attack framework. Harvey. Kamath. All for statistical estimation.

% Related work: Discuss the main approaches to lower bounds in DP. Packing works for pure DP, for Renyi DP (Nikolov++), and for eps,del up to a lower bound on sample complexity of $\log 1/\delta / \eps$. Fano. Local packing arguments [Asi-Duchi].

% Other hypothesis testing LB based approaches. [Audra]
% Assouad is perhaps best viewed as a product of packing arguments.

%  Discrepancy/Reconstruction approaches work for $(\eps,\delta)$ but do not give the right bound for large $\alpha$.

%  Fingerprinting... Information theoretic approaches : Salil's LB for Gaussians. Kramer-Rao type LB in our paper.

% Large amount of work on fingerprinting-based techniques. Ullman, ...

% (There are additional approaches to proving lower bounds such as those based on phrasing the problem as an optimization problem (staircase mechanism), Ramsay theory, that we will not get into here.)

\newpage

\thispagestyle{empty}

\setcounter{tocdepth}{2}
\tableofcontents

\clearpage
\pagenumbering{arabic}

\newpage

\section{Introduction}

Differential Privacy~\cite{DworkMNS06j} is standard notion of privacy in statistical databases. Differentially Private (DP) algorithms have been deployed by the US Census Bureau for publishing tables~\cite{AbowdACGHHJKLMMSSZ22}, the Israeli Department of Health for publishing birth records~\cite{HodC24}, and by several companies for training machine learning models and sharing statistics (see e.g.~\cite{Apple2017, DingKY17,xu2023federated, ZhangRXZZK23}). While there is now a large body of literature on designing differentially private algorithms for numerous machine learning and statistical tasks, algorithms for answering statistical (SQ) queries date back to some of the earliest work on Differential Privacy. In this work, we revisit the question of answering statistical/counting queries under a differential privacy constraint.

In a bit more detail, we consider a dataset of size $n$ where each user comes from a universe $\calX$. A counting query $q$ is defined by a function $p : \calX \rightarrow [-1,1]$ and the desired answer to such a query on a dataset $D = \{d_1,\ldots,d_n\}$ is $q(D) = \frac{1}{n} \sum_{i=1}^n p(d_i)$. We would like to design an $(\eps,\delta)$-DP algorithm that answers $Q$ counting queries, and we will measure the error of a mechanism $M$ on a dataset $D$ by either the  $\ell_\infty$ error $\max_{i \in Q} |M(D)_i - q_i(D)|$
or the (normalized) $\ell_2$ error $\sqrt{\frac{1}{Q} \sum_{i=1}^Q (M(D)_i -q_i(D))^2}$. There has been a large body of work on understanding upper and lower bounds for counting queries under differential privacy (see \cref{sec:related}).

Lower bounds for $(\eps,\delta)$-DP broadly fall into a few classes. Reconstruction arguments~\cite{DinurN03} show that an algorithm that is too accurate allows for reconstructing the input dataset. This approach was generalized by Muthukrishnan and Nikolov~\cite{MuthukrishnanN12} who recognized the hereditary discrepancy of the query set as the primary object of interest. Nikolov, Talwar and Zhang~\cite{NikolovTZ13} connected this to the geometry of the so-called {\em sensitivity polytope}, and showed that for {\em every} set of queries, the lower bound from hereditary discrepancy is within polylogarithmic factors of an appropriate Gaussian noise mechanism with a carefully chosen covariance, when $n$ is large enough. For the large dataset regime, this reconstrution/discrepancy approach yields nearly tight bounds not just in the worst case, but for any given set of queries.

For small datasets, this approach is suboptimal by polylogarithmic in $|\mathcal{X}|$ factors. Indeed, for a query set as simple as one-way marginals, the right achievable error rate was open until the seminal work of Bun, Ullman, and Vadhan~\cite{BunUV14}. This work showed that fingerprinting codes from cryptography can be adapted to prove tight lower bounds for certain sets of queries, including the tight bounds for one-way marginals. Further, they showed that these lower bounds can be ``composed'' with other lower bounds to get tight worst-case bounds. This approach has been used to prove worst-case lower bounds for several problems in differential privacy and in adaptive data analysis~\cite{HardtU14, SteinkeU15}. It is natural to ask if there is a geometric query-specific generalization of the fingerprinting codes technique.

In this work, we make progress towards pushing this technique to a broader class of problems. We develop a new abstract framework for proving lower bounds using the underlying approach in fingerprinting codes. Our approach has two main technical ingredients. First, we use an {\em exponential tilt} to define a family of distributions over an arbitrary set of points. This yields a family of distributions over the polytope $K$ of choice, that comes from the exponential family and is thus more amenable to the use of fingerprinting tools~\cite{KamathMS22}. Second, we relax the need for the ``parameter vector'' to come from a hyper-rectangle as in most previous work, using Stokes' theorem to control the appropriate ``score''. Recent work by Portella and Harvey~\cite{PortellaH24} used Stokes' theorem for a specific family of distributions via the Stein-Haff identity. Our work shows that this approach can be used to prove lower bounds for a large family of problems.

\subsection{Our Results}

% We summarize our results as the following.
%
% \subsubsection{Fingerprinting meets geometry}

%We develop a fingerprinting framework from a geometric perspective. From this new framework,
\paragraph*{Adaptive data analysis over a bounded domain.}
Adaptive Data Analysis addresses the question of statistical validity of query answers in the face of adaptivity. While $k$ non-adaptive counting queries can be answered to accuracy $\alpha$ using $O(\log k/\alpha^2)$ samples, these bounds break down when the queries can be adaptive. Dwork, Feldman, Hardt, Pitassi, Reingold and Roth~\cite{DworkFHPRR15} first showed that $O(\sqrt{k})$ samples suffice to answer $k$ queries. Using the private multiplicative weights algorithm of Hardt and Rothblum~\cite{HardtR10}, subsequent work by Bassily, Nissim, Smith, Steinke, Stemmer and Ullman~\cite{BassilyNSSSU16} showed that one can answer $Q$ statistical queries over a universe $\calX$ using a sample of size $O(\frac{\sqrt{\log(|X|)}\log(Q)}{\alpha^3})$. The currently known lower bounds of $\frac{\log Q}{\alpha^2}$ (folklore) and $\min(\sqrt{Q}, \sqrt{\log |\mathcal{X}|})/\alpha^2$ due to~\cite{HardtU14, SteinkeU15, NissimSSSU18} leave a significant gap.
Using our new framework, we strengthen known lower bounds on private query releasing and derive new results on adaptive data analysis, closing this gap up to a $\log \frac 1 \alpha$ factor.

 %We prove new lower bounds on adaptively answering statistical queries.
\begin{theorem}[Informal version of Theorem~\ref{theo:ada-lb-formal}]\label{theo:intro-ada-lb}
    Let $\calA$ be an algorithm for answering statistical linear queries over a domain $\calX$. Suppose $\calA$ operates on at most $o(\frac{\sqrt{\log|\calX|}\log(m)}{\alpha^3 \log(1/\alpha)})$ samples. Then, there is a $(\log |\calX|)$-round adaptive attack against $\calA$, which sends at most $m$ queries to $\calA$ and makes $\calA$ fail to be either sample- or distributional-accurate to within error $\alpha$ on at least one query. The attack succeeds with probability $\Omega(\alpha)$.
\end{theorem}

% To digest Theorem~\ref{theo:lb-ada}, we consider the parameter regime that is most interesting and relevant: namely, consider the case that $2^{m}\gg d$, $2^d \gg m$ and $\frac{1}{\alpha} \ll \min(2^d, 2^m)$. In this case, the domain size is understood as $|\mathcal{X}| = 2^{O(d)}$, and the number of queries is understood as $2^{O(m)}$. Then, Theorem~\ref{theo:lb-ada} says that any algorithm for answering adaptively generated statistical queries fails to be sample-accurate after answering $2^{O(m)}$ queries, so long as the algorithm only receives $n \le c\cdot \frac{m\sqrt{d}}{\alpha^3}$. Note that this is a lower bound of the form $\frac{\sqrt{\log |\mathcal{X}|}\cdot \log(\text{\# queries})}{\alpha^3}$, which is matched by a known upper bound implied by the celebrated private multiplicative weight mechanism.

\paragraph*{Removing the sample-accurate assumption.} Although all the state-of-the-art algorithms for adaptive data analysis offer accuracy with respect to both sample and distribution, by definition, an ADA algorithm does not need to be sample-accurate. For example, suppose one splits their data set into a couple of subsets and uses them in a sophisticated manner. It could be possible that somehow the outputs of the algorithm fail to be sample-accurate w.r.t.~the whole data set, but its outputs generalize nonetheless.

We prove the following lower bound against algorithms that are only distributional accurate. The bound we obtain is weaker than Theorem~\ref{theo:intro-ada-lb} by a factor of $\frac{1}{\alpha}$.
\begin{theorem}[Informal version of Theorem~\ref{theo:lb-for-all-algo-distribution}]\label{theo:intro-lb-for-distributional-only}
Let $\calX$ be the domain. Any ADA algorithm $\mathcal{A}$ over $\calX$ operating on $o(\frac{\sqrt{\log|\calX|}\log(Q)}{\alpha^2\log(1/\alpha)})$ samples cannot answer more than $Q$ adaptively generated statistical queries to within generalization error $\alpha$. Moreover, any such algorithm can be broken into $O(\log|\calX|)$ adaptive rounds.
\end{theorem}

Theorem~\ref{theo:intro-lb-for-distributional-only} still exhibits the right dependence on the universe size and the number of queries, but is off from the upper bound by a factor of $\frac{1}{\alpha}$. On the other hand, Theorem~\ref{theo:intro-ada-lb} shows that a nearly tight lower bound can be achieved by additionally assuming the algorithm is accurate w.r.t.~samples. Removing the ``sample-accurate'' assumption while retaining the correct dependence on $\alpha$ is an intriguing open question.

\paragraph*{Attack with bounded adaptivity.} The subtle distinction on whether the algorithm is sample-accurate or not has been thoroughly explored in the large-universe few-query regime. In particular, when $\mathcal{X}$ is unbounded, the Gaussian mechanism can use $\frac{\sqrt{k}}{\alpha^2}$ samples to answer $k$ adaptive queries to within error $\alpha$. On the lower bound side, following \cite{HardtU14}, the interactive fingerprinting code of \cite{SteinkeU15} proved a sample complexity of $\Omega\left(\frac{\sqrt{k}}{\alpha}\right)$. A subsequent work by Nissim, Smith, Steinke, Stemmer and Ullman~\cite{NissimSSSU18}  showed a stronger lower bound of $\Omega\left( \frac{\sqrt{k}}{\alpha^2} \right)$, \emph{assuming the algorithm is both sample-accurate and distributional-accurate}. It has been an important open question to close this gap and prove the same lower bound for algorithms only promised to be distributional accurate.

We make progress toward resolving this question by giving an attack with a bounded round of adaptivity. 
\begin{theorem}[Informal version of Theorem~\ref{theo:lb-for-fewer-queries}]\label{theo:intro-lb-for-fewer-queries}
Let $\calX$ be an unbounded domain. Any ADA algorithm $\mathcal{A}$ over $\calX$ operating on $o(\frac{\sqrt{k}}{\alpha^2\log(1/\alpha)})$ samples cannot answer more than $\frac{k}{\alpha^2}$ queries to within generalization error $\alpha$. Moreover, any such algorithm can be broken in $O(k)$ adaptive rounds.
\end{theorem}

In Theorem~\ref{theo:intro-lb-for-fewer-queries}, let $m = \frac{k}{\alpha^2}$ be the number of queries. In terms of $m$, this is still a $\frac{\sqrt{m}}{\alpha}$ lower bound. However, the upshot is that we can achieve the attack only with $k=O(\alpha^2 m)$ rounds of interaction, where each round sends $\frac{1}{\alpha^2}$ queries. If one can push our framework further by reducing the number of queries in each round to a constant, this would fully resolve the open question.

\paragraph*{Composition of query-releasing lower bounds.} The remarkable work of Bun, Ullman, and Vadhan~\cite{BunUV14} showed a nearly tight lower bound on the sample complexity for privately releasing arbitrary counting queries. Specifically, let $\calX$ denote the universe, $Q$ the number of queries asked, and $\alpha$ the desired $\ell_\infty$ accuracy parameter. It is shown that any $(\eps, \delta)$-DP algorithm must use $n$ samples where
$$
n \ge \Omega\left( \frac{\sqrt{\log |\calX|} \log Q}{\eps \alpha^2} \right).
$$
This is nearly tight, in the sense that there are known algorithms (\cite{HardtR10}) achieving:
$$
n \le O\left( \frac{\sqrt{\log|\calX| \log(1/\delta)} \log Q}{\eps \alpha^2} \right).
$$
A natural question is whether the additional $\log \frac 1 \delta$ factor is inherent. Given stronger lower bounds against pure-DP algorithms, the sample complexity must somehow depend on $\delta$. However, it was not known whether one can improve the sample complexity by ``decoupling'' the dependence on $\delta$ from every other factor.% Moreover, it is unclear if the lower bound holds for the more permissive $\ell_2$ error.

Via our framework, we bridge this final gap and show that this $\sqrt{\log(1/\delta)}$ multiplicative factor is inherent, for all reasonable range of $\delta$.
\begin{theorem}[Informal version of Theorem~\ref{theo:composition-lower-bound}]\label{theo:intro-composition-lb}
There is a workload matrix $A\in \{\pm 1\}^{Q\times \mathcal{X}}$ such that, for all reasonable choices of $0<\eps,\alpha,\delta<1$, any $(\eps,\delta)$-DP query releasing algorithm for $A$ requires 
$$
n\ge \Omega\left( \frac{\sqrt{\log(|\mathcal{X}|) \log(1/\delta)} \log(Q)}{\eps\alpha^2}\right)
$$
samples to answer $Q$ queries to within $\ell_\infty$ error of $\alpha$.
\end{theorem}
We can take $\delta \approx \frac{1}{\mathcal{X}}$ in Theorem~\ref{theo:intro-composition-lb}, and recover the best-known lower bounds against pure-DP algorithms \cite{Hardt11-thesis}.

A similar improvement of $\sqrt{\log(1/\delta)}$ can be shown for the class of two-way marginal queries. Like previous works, an average-error ($\ell_2^2$-metric) lower bound can be proved in this case.
\begin{theorem}[Informal version of Theorem~\ref{theo:two-way-marginal-lb}] \label{theo:intro-two-way-lb}
Let $A\in \{\pm 1\}^{d^2\times 2^d}$ be the two-way marginal query matrix on $d$ attributes. For all reasonable choices of $0 < \eps,\alpha,\delta < 1$, any $(\eps,\delta)$-DP query releasing algorithm for $A$ requires 
$$
n \ge \Omega\left( \frac{\sqrt{d \log(1/\delta)}}{\alpha^2 \eps} \right)
$$
samples to answer $d^2$ queries to within $\ell_2^2$ error of $\alpha^2 d^2$.
\end{theorem}

Theorem~\ref{theo:intro-two-way-lb} is also tight up to constants. For a matching upper bound, see, e.g.~\cite{DworkNT15}.%\cite{Cohen-AddaddEMZ24}.

\paragraph*{The sample complexity of privately releasing random linear queries.}
Suppose there are $N$ types of users (i.e., the universe is of size $N$). A workload of $d$ linear queries can be described by a matrix $A\in [-1,1]^{d\times N}$ where $A_{i,j}$ is the contribution of a type-$j$ user to the $i$-th query. The private query release problem is to publish the counting queries defined by the rows of $A$, up to a small $\ell_2$ error. It is easy to see that this problem is equivalent to outputting the mean of $n$ points in $\Re^d$, where each point is constrained to be one of the columns of $A$.

Nikolov, Talwar and Zhang~\cite{NikolovTZ13} studied this problem and showed that for every $A$, a carefully chosen Gaussian mechanism is within $\polylog(d)$ of the optimal when $n\ge d$. For smaller $n$, they showed that projecting the output of a Gaussian mechanism to the convex hull of the columns of $A$ is within $\polylog(d,N)$ of the optimal. Removing the $\polylog(N)$ dependence here was open, and it would have been reasonable to conjecture that the lower bound should be improved. Using our framework, we show that for random $A \in \{-1,1\}^{d\times N}$, the lower bound can indeed be improved for a range of small-error settings.

\begin{theorem}[Informal version of Theorem~\ref{theo:random-query-lb}]\label{theo:intro-random-lb}
Let $N,d\ge 0$ be two integers such that $d\ll N\ll 2^d$. With probability $1-o(1)$ over a random matrix $A\in\{\pm 1\}^{d\times N}$, the following is true: for all $\alpha < o\left( \frac{\sqrt{\log(N)}}{\sqrt{d}} \right)$, any approximate-DP algorithm $\calA$ for query releasing on workload matrix $A$ needs at least $\Omega\left( \frac{\sqrt{d}}{\alpha} \right)$ samples to achieve $\ell_2^2$-error of $\alpha^2 d$.
\end{theorem}

Theorem~\ref{theo:intro-random-lb} discusses the complexity of query releasing from a sample-complexity perspective. Equivalently, we can approach the problem from an error-complexity perspective and write Theorem~\ref{theo:intro-random-lb} in an equivalently way: any private algorithm with $n\ge \frac{d}{\sqrt{\log N}}$ samples incurs an error of $\alpha \ge \Omega(\frac{\sqrt{d}}{n})$, implying that the Gaussian mechanism is optimal when $n\ge \frac{d}{\sqrt{\log N}}$. Previously, a lower bound of the form $\frac{\sqrt{d}}{n}$ was only known for $n \ge d$ via hereditary discrepancy \cite{MuthukrishnanN12, NikolovTZ13}.

For a much smaller $n$, we do not expect $\frac{\sqrt{d}}{n}$ to be a lower bound. Indeed, the projection mechanism of \cite{NikolovTZ13} gives an upper bound of $\tilde{O}(\frac{(\log N)^{1/4}}{\sqrt{n}})$ for the whole range of $n$, which outperforms the Gaussian mechanism (whose error is $\frac{\sqrt{d}}{n}$) when $n \le \frac{d}{\sqrt{\log N}}$. In light of Theorem~\ref{theo:intro-random-lb}, the threshold $\frac{d}{\sqrt{\log N}}$ does not appear to be coincidental: it is reasonable to conjecture that $\Omega(\frac{(\log N)^{1/4}}{\sqrt{n}})$ serves as a lower bound for $n\le \frac{d}{\sqrt{\log N}}$, meeting with the $\Omega(\frac{\sqrt{d}}{n})$ lower bound at $\frac{d}{\sqrt{\log N}}$. For some classes of structured query sets (e.g., two-way marginals), the conjecture has been affirmatively confirmed \cite{BunUV14}. See also Theorem~\ref{theo:intro-two-way-lb} and the right part of Figure~\ref{figure:compare-random-with-worst}.
%TODO: define the problem of private query release. Mention that our error metric of interest is $\ell_2$ metric.

Surprisingly, though, we show that it is the \emph{upper bound} that can be improved for smaller $n$ (corresponding to larger $\alpha)$. In particular, for all $n < \frac{d}{\log N}$, there is an error upper bound of $O(\frac{1}{\sqrt{n}})$.
% , we show that it is the upper bound that can be improved for larger $\alpha$. Indeed, for random queries, the projection mechanism is never optimal.
%We show the projection mechanism is not optimal for a random query matrix $A$.

\begin{theorem}[Informal version of Theorem~\ref{theo:query-release-via-histogram}]\label{theo:intro-random-ub}
Let $N,d\ge 0$ be two integers such that $d\ll N\ll 2^d$. With probability $1-o(1)$ over a random matrix $A\in\{\pm 1\}^{d\times N}$ the following is true: for all $\alpha \in \left(\omega(\frac{\sqrt{\log(N)}}{\sqrt{d}}), 1\right)$, there is an approximate-DP algorithm for query releasing on workload $A$. The algorithm uses $O(\frac{1}{\alpha^2})$ samples and achieves $\ell_2^2$ error of $\alpha^2 d$.
\end{theorem}

\begin{figure}[t]
  \begin{subfigure}{0.4\textwidth}
 % \resizebox{\textwidth}{!}{%
    \centering
\begin{tikzpicture}
%
  % Define axis labels
  \draw[->] (0,0)--(5,0);
  \node[below] at (2.5,-1) {Inverse Accuracy  $\alpha^{-1}$};
  \draw[->] (0,0)--(0,4);
  \node[above, rotate=90] at (-0.5,2) {Sample Complexity $n$};
%
  % Label points of interest
  \draw[-,very thin] (0.05,-0.05)--(0.05,0.05);
  \node[left] at (0.5,-0.5) {\small{$\alpha=1$}};
  \draw[-,very thin] (2,-0.05)--(2,0.05);
  \node[left] at (3.5,-0.5) {\small{$\alpha = \frac{\sqrt{\log |\calX|}}{\sqrt{d}}$}};
%
  % Draw the two lines
  \draw[black] (0,0)--(2,2);
  \draw[black] (2,3)--(4,4);
 \draw[->, blue] (3.5,1.5) -- (1.6, 1.5);
   \node[right] at (3.5, 1.5)  {\small{$n = \tilde{\Theta}(\frac{1}{\alpha^2})$}};
 \draw[->, blue] (3.5,3) -- (3,3.4);
  \node[right] at (3.5,3) {\small{$n = \tilde{\Theta}(\frac{\sqrt{d}}{\alpha})$}};
  \end{tikzpicture}
 % }
\end{subfigure}
 \hspace{0.15\textwidth}
\begin{subfigure}{0.4\textwidth}
 %  \resizebox{\textwidth}{!}{%
   \centering
\begin{tikzpicture}
  % Define axis labels
  \draw[->] (0,0)--(5,0);
  \node[below] at (2.5,-1) {Inverse Accuracy  $\alpha^{-1}$};
  \draw[->] (0,0)--(0,4);
  \node[above, rotate=90] at (-0.5,2) {Sample Complexity $n$};
%
  % Label points of interest
  \draw[-,very thin] (0.05,-0.05)--(0.05,0.05);
  \node[left] at (0.5,-0.5) {\small{$\alpha=1$}};
  \draw[-,very thin] (2,-0.05)--(2,0.05);
  \node[left] at (3.5,-0.5) {\small{$\alpha = \frac{\sqrt{\log |\calX|}}{\sqrt{d}}$}};
%
  % Draw the two lines
  \draw[black] (0,1)--(2,3);
  \draw[black] (2,3)--(4,4);
 \draw[->, blue] (3.5,1.5) -- (1.6, 2.5);
   \node[right] at (3.5, 1.5)  {\small{$n = \tilde{\Theta}(\frac{\sqrt{\log |\calX|}}{\alpha^2})$}};
 \draw[->, blue] (3.5,3) -- (3,3.4);
  \node[right] at (3.5,3) {\small{$n = \tilde{\Theta}(\frac{\sqrt{d}}{\alpha})$}};
    \end{tikzpicture}
%}
\end{subfigure}
       %  \captionsetup{justification=centering}
    \caption{Behavior of sample complexity vs.~error trade-off for $d$ random linear queries (left) and worst-case queries (right) over a universe $\calX$ ($\log$-$\log$ scale). The sample complexity for random queries is discontinuous at $\alpha \approx \frac{\sqrt{\log |\calX|}}{\sqrt{d}}$. The dependence on the privacy parameters and $\log d$ terms are suppressed for clarity.}
    \label{figure:compare-random-with-worst}
\end{figure}
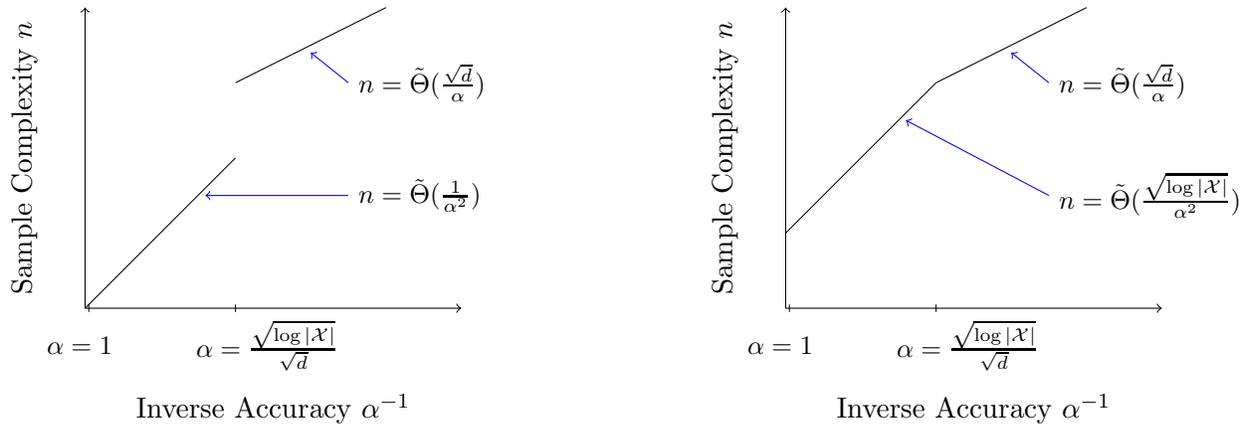

For $n < d$, the reconstruction argument of \cite{DinurN03,MuthukrishnanN12} gives an error lower bound of $\frac{1}{\sqrt{n}}$. Combining Theorems~\ref{theo:intro-random-lb}, \ref{theo:intro-random-ub} and known results, we have thus completely determined the sample-error trade-off for random queries (see left part of Figure~\ref{figure:compare-random-with-worst}). This shows a surprising discontinuity in the sample complexity as a function of $\alpha$. Notably, for the closely related notion of concentrated DP, existing bounds~\cite{BlasiokBNS19} show that a dependency on $\log |\calX|$ is necessary even for large $\alpha$. This also shows that (once again in contrast to other privacy notions), random queries are not the worst case: as mentioned above, for worst-case queries the $\sqrt{\log |\calX|}\log Q / \alpha^2$ dependence holds for essentially the whole range of $\alpha$.
%
% One might ask whether Theorem~\ref{theo:query-release-ub} is tight. [TODO: review previously known lower bounds, and mention there is a gap in a certain parameter regime]. We can use our fingerprinting framework to bridge this gap.
%
% \begin{theorem}
%     TODO: state the fingerprinting lower bound here.
% \end{theorem}
%
% [TODO: Add a discussion of this proof approach. What was in BUV. What was clarified in Tracing from Trace Amounts paper. What was done in the Score attack framework paper. What was done in the Nick Harvey paper. Also mention recent Jon Ullman paper. For adaptive data analysis, discuss what was done in the Hardt-Ullman paper.]
%

\paragraph*{Future directions.} Our work leaves several open research directions. We have demonstrated that this general framework can help prove new lower bounds in differential privacy and adaptive data analysis. %It is natural to investigate if the lower bounds hold against {\em balanced} adversaries.
Our work shows that fingerprinting tools can allow us to extend instance-specific lower bounds to a wider range of parameters for random queries. However, the low-accuracy regime can sometimes admit new algorithms, and we leave to future work a better understanding of the geometric properties that determine the sample complexity of a given query set. Adaptivity is a challenge beyond data analysis, e.g. in streaming algorithms~\cite{BenEliezerJWY22,HasidimKMMS20}, sampling~\cite{BenEliezerY20} and dynamic data structures~\cite{BeimelKMNSS22} and some of the tools developed in our work may help prove lower bounds for adaptivity in those settings as well.

\subsection{Related Work}
\label{sec:related}

Fingerprinting codes were proposed in cryptography by Boneh and Shaw~\cite{BonehS98}. Tardos~\cite{Tardos03} gave an optimal construction of these objects. Dwork, Naor, Reingold, Rothblum and Vadhan~\cite{DworkNRRV09} first used cryptographic traitor tracing schemes to prove lower bounds in DP, and Bun, Ullman and Vadhan~\cite{BunUV14} showed that information-theoretic fingerprinting codes constructions and their extensions imply strong lower bounds for query release under differential privacy. The ``fingerprinting lemma'' from~\cite{BunUV14} has been used in several lower bounds for other problems~\cite{DworkTTZ14, BassilyST14, SteinkeU16,SteinkeU17,NarayananME22,KamathMS22,Narayanan23, CaiWZ23,  PortellaH24, PeterTU24}. Dwork, Smith, Steinke, Ullman and Vadhan~\cite{DworkSSUV15} simplified the analysis of the fingerprinting attacks and gave a first-principles proof. Recent work by Cai, Wang and Zhang~\cite{CaiWZ23} generalized and formalized these ``score attacks'' for parameter estimation for a family of distributions, but their approach is still restricted to hyperrectangles (see~\cite{PortellaH24}). As discussed above, Kamath, Mouzakis and Singhal~\cite{KamathMS22} extended fingerprinting attacks to exponential families over hyper-rectangles, and Portella and Harvey~\cite{PortellaH24} first showed how to go beyond i.i.d. parameter distributions by using Stokes' theorem variants. Narayanan~\cite{Narayanan23} used the score attack framework to improve covariance lower bounds, and Peters, Tzafidia and Ullman~\cite{PeterTU24} extend the fingerprinting lemma to prove lower bounds in the weak accuracy regime.

Other general approaches to proving lower bounds include reconstruction~\cite{DinurN03, DworkMT07, DworkY08} and discrepancy approaches~\cite{MuthukrishnanN12, NikolovTZ13}; and information theoretic techniques proposed by Acharya, Sun and Zhang~\cite{AcharyaSZ21} that have been used in some recent works~\cite{KamathMS22, FeldmanMST24}.
%Fingerprinting line of work. Discuss Score attack, ...

The query release problem has been studied in many works, starting with Blum, Dwork, McSherry and Nissim~\cite{BlumDMN05}. This problem was first studied from a geometric instance-optimality viewpoint by Hardt and Talwar~\cite{HardtT10} for pure DP; by Nikolov, Talwar and Zhang~\cite{NikolovTZ13} for approximate DP; and by Blasiok, Bun, Nikolov and Steinke~\cite{BlasiokBNS19}  for Concentrated DP~\cite{DworkR16, BunS16CDP}. While these results are largely for $\ell_2$ error, bounds for $\ell_2$ can often be converted to those for $\ell_\infty$ by private boosting~\cite{DworkRV10}. In some cases~\cite{BlasiokBNS19}, instance-dependent bounds for $\ell_\infty$ can be proven by more direct means. The work of Blum, Ligget and Roth~\cite{BlumLR13} first showed that the sample complexity in the low-accuracy regime can behave differently from that in the high-accuracy regime, and subsequent work made these results more efficient and extended them to approximate differential privacy~\cite{RothR10, HardtR10}.

The use of differentially private algorithms for Adaptive Data Analysis was initiated in the seminal work of Dwork, Feldman, Hardt, Pitassi, Reingold, and Roth~\cite{DworkFHPRR15} and the aforementioned bounds for private multiplicative weights are from Bassily, Nissim, Stemmer, Steinke and Ullman~\cite{BassilyNSSSU16}. This has triggered a long line of research on adaptive data analysis~\cite{DworkFHPRR15b, DworkFHPRR15c, RogersRST16, RussoZ16,Smith17, FeldmanS17, FeldmanS18, NissimSSSU18, ShenfeldL19, JungLNRSS21, FishRR20, DaganK22, KontorovichSS22, DinurSWZ23, Cohen0NSS23, Blanc24}. There has been a beautiful line of work on lower bounds for adaptive data analysis by Hardt and Ullman~\cite{HardtU14}, Steinke and Ullman~\cite{SteinkeU15}. In addition to the aforementioned information-theoretic results, this line of work shows that for computationally bounded algorithms in the large $\calX$ regime, the sample complexity is $\Omega(\sqrt{Q})$. These results have been extended to apply to hold against a more restricted class of adversaries in~\cite{Elder16, NissimST24}. There have also been many works that have aimed to replace the strong stability notion in differential privacy by weaker notions, and our lower bounds imply that these cannot asymptotically improve on the results from private multiplicative weights in the general case.

\section{Fingerprinting Codes Meet Geometry}\label{sec:fingerprinting-meets-geometry}

\paragraph*{Notation.} We will work in the $d$-dimensional Euclidean space $\mathbb{R}^d$. For a vector $v\in \mathbb{R}^{d}$, we use $v_i$ to denote the $i$-th coordinate of $v$. We will let $v_{\le r}$ denote the first $r$ coordinates of $v$, and let $v_{-i}$ denote all but the $i$-th coordinate. Similarly we use $v_{>r}, v_{< r}$ etc. For a list of vectors $v^1,\dots, v^m$, we use superscript (e.g., $v^j$) to index individual vectors. We write $a\circ b$ to denote the concatenation of two objects in a natural manner (be it lists, vectors, or Boolean strings, etc.). For a distribution $\calD$ over $\mathbb{R}^d$, the \emph{covariance matrix} of $\calD$ is defined as
$$
\COV(\calD) \coloneqq \E_{v\sim \calD}[ (v - \E[\calD]) (v - \E[\calD])^\top ].
$$
A zero-mean random variable $A$ is \emph{$K$-subgaussian} if $\|A\|_p \le K \sqrt{p}$ for every $p\ge 1$. Equivalently, $A$ is $O(K)$-subgaussian if $\Pr[|A|>Kt]\le 2 \exp(-t^2)$ for every $t \ge 1$. We say two random variables $X,Y$ are \emph{$(\eps,\delta)$-indistinguishable}, if $\Pr[X\in S] \in [e^{-\eps}(\Pr[Y\in S] - \delta), e^\eps \Pr[Y\in S] + \delta]$ for every measurable set $S$.

We use the following convention to ease our asymptotic analysis: We frequently write $c$ (resp.~$C$) to denote a small (resp.~large) but absolute constant. Their appearance in different contexts might have different values. Generally, when we write a statement like ``for $n \ge c\cdot m$, something is true'', what we mean is that there exists an absolute constant $c > 0$, such that with $n \ge c\cdot m$, the said statement is true.

\subsection{Geometry Preliminaries}

\newcommand{\diver}{\mathrm{div}}
\newcommand{\Var}{\mathrm{Var}}

Let $S$ be a closed surface in the $d$-dimensional Euclidean space, and $V$ be the region enclosed by $S$. Let $f:\mathbb{R}^d\to \mathbb{R}^d$ be a vector field. The {\em divergence} of $f$ at a point $x$ is defined as
$$
\diver f(x) = \sum_{i=1}^{d} \frac{\partial}{\partial x_i}f(x)_i.
$$
The divergence theorem then says that
$$
\iiint_{V} \diver f(x) dV = \oiint_S \langle f(x), \vec{n}\rangle dS.
$$
Here, $\vec{n}$ denotes the normal vector of $S$ at a point $\theta\in S$. It might be helpful to review some typical examples of the divergence theorem:
\begin{itemize}

    \item A simple example is where $V$ is the $\ell_{\infty}$ ball of coordinate-wise bound $R/2$ (namely, a hypercube centered at the origin with side length $R$), and $S$ its boundary (union of $2d$ faces). Since $\frac{\Vol(V)}{\Area(S)} = \frac{R^d}{2d R^{d-1}} = \frac{R}{2d}$, we obtain
    $$
    \text{Average divergence inside $V$} = \frac{2d}{R} \cdot \text{Average of $\langle f, \vec{n}\rangle$ on the surface}.
    $$
    The divergence theorem is intuitive here: since $V$ is a product region, the ``contribution'' to the divergence integral from the $d$ coordinates is independent, allowing us to integrate each $i\in [d]$ separately. Each $i\in [d]$ is naturally associated with a pair of faces of $S$, namely $\{x:\|x\|_\infty \le R/2, x_i = \pm R/2\}$. To compare the integrals on two sides, we can condition on $x_{-i}$ and apply the fundamental theorem of calculus on $x_i$. In this proof, no knowledge of multi-variate calculus is required.
    \item A more interesting case is that of $V$ being an $\ell_2$-ball of radius $R$ and $S$ its enclosing sphere. In this case, one can verify that $\frac{\Vol(V)}{\Area(S)} = \frac{R}{d}$. Hence,
    $$
    \text{Average divergence inside $V$} = \frac{d}{R} \cdot \text{Average of $\langle f,\vec{n}\rangle$ on the surface}.
    $$
\end{itemize}

\subsection{The Exponential Family}

The first step of the fingerprinting argument is designing a suitable family of distributions. We will work with a family of distributions known as ``exponential family'', defined as follows.

Let $K\subseteq \mathbb{R}^d$ be a set of vectors. We consider $K$ as the possible inputs to the private algorithm and define our ``base distribution'' as the uniform distribution over $K$. In this work, we will largely be concerned with the case of finite $K$ though the approach extends easily to the case of $K$ being a bounded convex set, say, or even more broadly to an arbitrary base distribution. For every $\theta \in \mathbb{R}^d$, we define an $\theta$-tilt distribution $D_{\theta}(K)$ over $K$ as:
$$
\Pr_{x\sim D_{\theta}}[x = v] \propto \exp(\theta^\top x).
$$
Intuitively, compared with the uniform distribution over $K$, we put more ``favor'' on points with a large inner product with $\theta$. When the base set $K$ is clear from context, we will use $D_{\theta}$ to denote $D_{\theta}(K)$ for brevity.

Now, let $\calA:K^n\to \mathbb{R}^{d}$ be an algorithm that receives $n$ samples from $K$ and returns approximately their average. We would like to understand the privacy-utility tradeoff of such algorithms. For every $\theta$, we may define
$$
g(\theta) \coloneqq \mathbb{E}_{v^1,\dots, v^n\sim D_{\theta}^n}[\calA(v)]
$$
to be the average output of $\calA$ on a dataset drawn from $D_{\theta}^n$. We also define
$$
\mu_\theta = \mu(\theta) \coloneqq \mathbb{E}_{v\sim D_{\theta}^n} [v]
$$
as the ``true mean'' of the distribution $D_{\theta}$. Note that both $g(\theta)$ and $\mu(\theta)$ are understood as vector fields over $\mathbb{R}^d$. For any not-too-small $n$, a generalization argument implies that $\mu(\theta)$ and $g(\theta)$ are close, assuming the algorithm $\calA$ is accurate w.r.t.~samples.

\paragraph*{Defining the ``score''.} The reason we work with exponential families is crystallized in the following proposition, which says that the derivative of $g(\theta)$ with respect to $\theta$ is related to how well $\calA(v)$ correlates with $v$. A version of this proposition appeared in \cite{KamathMS22}.

\begin{proposition}\label{prop:differentiate-to-score}
    For every $i\in [d]$, we have
    $$
    \frac{\partial}{\partial \theta_i} g(\theta)_i = \mathbb{E}_{v^1,\dots, v^n \sim D_{\theta}^n } \left[ \calA(v)_i \cdot \left[ \sum_{j=1}^{n} v^j_i - (\mu_\theta)_i \right]  \right].
    $$
    Summing up all $i\in [d]$, we obtain
    $$
    \diver g(\theta) = \mathbb{E}_{v^1,\dots, v^n \sim D_{\theta}^n}\left[ \left\langle \calA(v), \sum_{j=1}^{n} v^j - \mu_\theta \right\rangle \right].
    $$
\end{proposition}

\begin{proof}
We use $v^{-j}$ to denote a list of $(v^1,\dots, v^{j-1}, v^{j+1},\dots, v^n)$. Using the chain rule of calculus, we consider the ``impact'' of differentiating $\theta$ to each $v^j\sim D_\theta$. Namely, we have

\begin{align}
    \frac{\partial}{\partial \theta_i} g(\theta)_i
    &= \sum_{j=1}^{n} \E_{v^{-j} \sim D_{\theta}^{n-1}} \left( \frac{\partial}{\partial \theta_i} \left( \E_{v^{j}\sim D_{\theta}} [\calA(v)_i] \right) \right) \notag \\
    &= \sum_{j=1}^{n} \E_{v^{-j} \sim D_{\theta}^{n-1}} \left( \sum_{u\in K} \frac{\partial \Pr[D_{\theta} = u]}{\partial \theta_i} \calA(v^{-j} \circ u)_i\right) \label{equ:before-differentiate-Pr}
\end{align}
We calculate
$$
\begin{aligned}
    \frac{\partial \Pr[D_{\theta} = u]}{\partial \theta_i}
    &= \frac{\partial }{\partial \theta_i}\left( \frac{\exp(\langle u, \theta\rangle)}{\sum_{u'} \exp(\langle u', \theta\rangle)} \right) \\
    &= \frac{\exp(\langle u, \theta\rangle)\cdot u_i}{\sum_{u'} \exp(\langle u', \theta\rangle)} - \frac{\exp(\langle u, \theta\rangle) \sum_{u'} \exp(\langle u', \theta\rangle) \cdot u'_i}{\left( \sum_{u'} \exp(\langle u', \theta\rangle) \right)^2} \\
    &= \Pr[D_\theta = u]\cdot u_i - \Pr[D_{\theta} = u] \cdot \E_{u'\sim D_{\theta}}[ u'_i ] \\
    &= \Pr[D_{\theta} = u] \left( u_i - (\mu_{\theta})_i \right).
\end{aligned}
$$
Back to the derivation before, we proceed as
$$
\begin{aligned}
    \eqref{equ:before-differentiate-Pr}
    &= \sum_{j=1}^{n} \E_{v^{-j} \sim D_{\theta}^{n-1}} \left( \sum_{u\in K} \Pr[D_{\theta} = u] \calA(v^{-j} \circ u)_i \cdot (u_i - (\mu_{\theta}))_i \right)  \\
    &= \sum_{j=1}^{n} \E_{v^{-j}} \E_{v^j} \left[ \calA(v)_i \left( v^j_i - (\mu_\theta)_i \right) \right] \\
    &= \E_{v} \left[ \calA(v)_i \left( \sum_{j=1}^{n} v^j_i - (\mu_\theta)_i \right) \right].
\end{aligned}
$$
This completes the proof of the first equation. For the second equation, we simply sum up all $i$'s and appeal to the definition of divergence.
\end{proof}
A more general formulation of this exists, that holds beyond exponential families. Indeed the term $(v^j - \mu_\theta)$ can in general be replaced by the $S_{\theta}(v^j) = \nabla_{\theta} \log p_{D_{\theta}}[v^j]$ (see~\cite{CaiWZ23}). Since this generality will not be needed in our work, we restrict our attention to exponential families.

Looking ahead, we will define $\langle v^j - \mu_\theta, \calA(v)\rangle$ as the {\em score} of $v^j$ with respect to the algorithm's output $\calA(v)$. We prove our desired lower bounds by deriving contradicting upper and lower bounds on the score by utilizing the privacy and accuracy guarantee of the algorithm, respectively.

\subsection{Template Overview}\label{sec:fingerprinting-overview}

In this subsection, we briefly introduce how to derive upper and lower bounds on the {\em score} of inputs.

\paragraph*{On upper bounding the score.} Given the output of the algorithm $\calA(v)$, it is unlikely to correlate well with a freshly sampled $v'\sim D_{\theta}$ (since the algorithm has never seen $v'$ before). On the other hand, due to the privacy property of $\calA$, its output is unlikely to change too much when we replace one input $v^j$ with $v'$. These two facts combined would give an upper bound on the score of in-sample data points.

We establish the following proposition, which would imply that the ``score'' of an in-sample point $v^j$ cannot be significantly larger than that of an independent sample.

\begin{proposition}\label{prop:privacy-imply-small-score-difference}
    Suppose $X, X'$ are a pair of $(\varepsilon, \delta)$-indistinguishable random variables supported on $[-B, B]$ such that $\Var[X] = b^2$ and $\mathbb{E}[X] = 0$. Then,
    $$
    \mathbb{E}[X'] \le O((e^\eps - 1) b + B \delta).
    $$
\end{proposition}

\begin{proof}
     We pay a price of $B \delta$ to ``change'' $X'$ into a random variable that is $(\varepsilon, 0)$-indistinguishable with $X$. For this new $X'$, using integration by parts we get
     $$
     \begin{aligned}
     \mathbb{E}[X']
     & = \mathbb{E}[X' - X] \\
     & = \int_{0}^B \left( \Pr[X'\ge t] - \Pr[X \ge t] \right) dt + \int_{-B}^0 \left(\Pr[X\le  t] - \Pr[X' \le t]  \right)dt \\
     & \preceq (e^\eps - 1) + \int_{-B}^{- b}  \left(\Pr[X\le  t] - \Pr[X' \le t]  \right)dt  + \int_{ b}^B \left( \Pr[X'\ge t] - \Pr[X \ge t] \right) dt  \\
     &\preceq (e^\eps - 1) + \int_{-B}^{- b}  \frac{\varepsilon b^2}{t^2}dt + \int_{ b}^B \frac{\varepsilon  b^2}{t^2} dt \\
     &\le O((e^\eps - 1) b).
     \end{aligned}
     $$
     To justify the derivation, the first inequality used the fact $\Pr[X'\ge t] - \Pr[X\ge t] \le (e^\eps - 1)$ for every $t$, and the second inequality used Chebyshev's inequality. This completes the proof.
\end{proof}

We will use Proposition~\ref{prop:privacy-imply-small-score-difference} to compare $\langle v^j - \mu_\theta, \calA(v)\rangle$ with $\langle v' - \mu_\theta, \calA(v)\rangle$ where $v'\sim D_\theta$ is independently drawn. These two random variables are $(\eps,\delta)$-indistinguishable by the privacy property of $\calA$. Also, it is easily seen that the latter random variable has zero mean. Therefore, to use Proposition~\ref{prop:privacy-imply-small-score-difference}, it remains to upper bound the variance (equivalently, the second moment) of $\langle v' - \mu_\theta, \calA(v)\rangle$. One way to proceed is the following: Conditioning on $\calA(v)$, we have
$$
\E[\langle v' - \mu_\theta, \calA(v) \rangle^2] = \calA(v)^\top \COV(D_{\theta}) \calA(v).
$$
Therefore, it suffices to upper bound the spectral norm of the covariance of $D_\theta$. The way we accomplish this will depend on the geometry of $K$ and how it interplays with the exponential tilt. The details are deferred to future application sections.

\paragraph*{On lower bounding the score.} Roughly, we will lower bound $\langle v^j-\mu_\theta, \calA(v)\rangle$ by utilizing Proposition~\ref{prop:differentiate-to-score} together with the accuracy property of $\calA$. The first step is to consider a randomly chosen $\theta$ from a region $V$ and relate the score with an integral of divergence (by Proposition~\ref{prop:differentiate-to-score}). The latter is further related to a surface integral (on the surface $S$ that encloses $V$) by the divergence theorem. Namely, for a random $\theta$ drawn uniformly from a closed region $V$, we have
\begin{align}
\E_{\theta \sim V} \left( \E_{v\sim D_\theta^n} \left[ \sum_{j} \langle \calA(v), v^j - \mu_\theta \rangle \right] \right)
& = \frac{1}{\Vol(V)} \iiint_V \diver g(\theta) d\theta & \text{(Proposition~\ref{prop:differentiate-to-score})} \notag \\
&= \frac{1}{\Vol(V)} \oiint_S \langle g(\theta), \vec{n}_\theta\rangle d\theta & \text{(the divergence theorem)} \notag \\
&= \frac{\Area(S)}{\Vol(V)} \E_{\theta \sim S}[ \langle g(\theta), \vec{n}_\theta \rangle ]. \label{equ:IP-to-be-lbed}
\end{align}
Here, we use $\vec{n}_\theta$ to denote the normal vector of $S$ at the point $\theta$.

We further lower bound \eqref{equ:IP-to-be-lbed} in an application-specific way. We will always replace $g(\theta)$ in \eqref{equ:IP-to-be-lbed} with $\mu(\theta)$ (this is possible because the algorithm is assumed to be accurate, which means the difference between $g(\theta)$ and $\mu(\theta)$ is minor). Finally, by utilizing the definition of $D_{\theta}$ and choosing an appropriate $S$, we can give a desired lower bound on $\langle \mu(\theta), \vec{n}_\theta \rangle$. For example, by choosing $V$ to be an $\ell_2$ ball and $S$ be its boundary (i.e.,~a sphere), we get that $\vec{n}_\theta = \frac{\theta}{\|\theta_2\|}$. Since $D_{\theta}$ is defined by favoring points that have a large inner product with $\theta$, we conclude that $\langle \mu(\theta), \vec{n}_\theta\rangle = \frac{1}{\|\theta\|_2} \langle \E_{v\sim D_\theta}[v], \theta\rangle$ is large as well.

Thus the framework needs only a few application-specific ingredients. The set $K$ and the body $V$ that $\theta$ lies in will depend on the application. We will need to prove for each application an upper bound on the spectral norm of the covariance $\COV(D_{\theta})$, and prove the lower bound on the $\E_{\theta \in S}[\langle \mu(\theta), \vec{n}_\theta \rangle]$. In some of our applications, it will be convenient to deviate slightly from this general framework, in which case we may need to redo the proofs of slight variants of some of the steps in this general recipe.

% \paragraph*{Template overview.} We view the problem from a sample complexity perspective. That is, we fix the accuracy $\alpha$ and privacy $\varepsilon$ parameters and work out a lower bound on the number of samples, which we denote as $n$.

% The privacy definition is the standard item-level one. The accuracy requirement for the algorithm is that
% $$
% \| g(\theta) - \mu_\theta \|_2 \le \alpha
% $$
% for every $\theta$. Note that this requires the algorithm's answer to be close to the population mean instead of the sample mean. However, these two conditions are equivalent once we have a moderate number of samples.

% In general, the lower bound proof consists of the following steps:
% \begin{itemize}
%     \item Choose a surface $S$ and its enclosed region $V$.
%     \item Lower bound the surface integral $\oiint_{S} g(\theta) \cdot \vec{n} dS$, which in turn lower-bounds the average divergence of the algorithm over $V$.
%     \item Upper bound the average divergence over $V$ by using the privacy property and reasoning about the variance of the ``score''.
% \end{itemize}

% Combining the bounds we get in Steps 2 and 3 will give a lower bound on the number of samples needed.

% \paragraph*{Known upper bound.} For reference, we note that the projection mechanism of [Nikolov-Talwar-Zhang], in our language and scaling, gives an upper bound of $n\le O(\frac{GW(K)}{\alpha^2 \varepsilon})$. The straightforward bound by Gaussian mechanism is $n\le O(\frac{\sqrt{d}}{\alpha \varepsilon})$, when suppressing the dependence on $\delta$.

\subsection{Proof for the Hypercube} \label{sec:fingerprinting-hypercube}

In this subsection, we (re-)prove the lower bound for answering one-way marginal queries using our framework. We hope it serves as a warm-up to the more complicated applications later on.

We work with an equivalent formulation of the problem, which is the task of releasing the mean of $n$ vectors from a $d$-dimensional Boolean hypercube. In this case, we have the set $K = \{ \pm 1 \}^d$. Let $\eps < \frac{1}{10}$. We aim to prove that there does not exist an $(\eps,\delta)$-DP algorithm which, on input $n\le c\cdot \frac{\sqrt{d}}{\eps}$ (for some small absolute $c > 0$) Boolean vectors $x^1,\dots, x^n\in K$, with probability $1$ returns a vector $\tilde{x}$ such that $\| \tilde{x} - \frac{1}{n} \sum_{j} x^j \|_2 \le \frac{1}{100} \sqrt{d}$. Namely, the algorithm makes constant error per query on average.

Assume for contradiction that such an algorithm $\calA$ exists. We give contradicting upper and lower bounds on the ``score'' of the inputs to $\calA$. Let $V$ be an $\ell_2$-ball of radius $R = \Theta(\sqrt{d})$. We consider the exponential family $\{D_\theta\}$ on $K$ parameterized by $\theta\in V$.

\paragraph*{Upper bound the score.} Let $\theta \in V$ be arbitrary. Conditioning on $x^1,\dots, x^n\sim D_\theta^n$ and the output $\calA(x)$, it is easy to see that
$$
\Var_{v'\sim D_{\theta}}\left[ \langle v' - \mu_\theta , \calA(x) \rangle \right] = \sum_{i=1}^{n} \Var_{v'} [ (v'_i - (\mu_\theta)_i ) \cdot \calA(x)_i].
$$
This is because, due to the structure of $K$ and the definition of $D_{\theta}$, different coordinates of $v'$ are independent (basically, we have $\Pr[v'_i = +1] = \frac{\exp(\theta_i)}{\exp(\theta_i) + \exp(-\theta_i)}$ independently for every $i$).

Note that we can assume $\calA(x)_i \in [-1, 1]$ (if not, we can truncate it into this range without increasing error). Now, $(v'_i - (\mu_\theta)_i ) \cdot \calA(x)_i$ is a bounded random variable, which has variance at most $1$. Hence, we obtain
$$
\Var_{v'\sim D_{\theta}}\left[ \langle v' - \mu_\theta , \calA(x) \rangle \right] \le d.
$$
By Proposition~\ref{prop:privacy-imply-small-score-difference}, this implies that for every $v^j$, it holds
\begin{align}
\E_{v\sim D_\theta^n,\calA(v)}[\langle \calA(v), v^j - \mu_\theta \rangle] \le O((e^\eps - 1)\sqrt{d} + \delta d) = O(\eps \sqrt{d} + \delta d). \label{equ:hypercube-ub-finish}
\end{align}
This bound holds for every fixed $\theta$.

\paragraph*{Lower bound.} To establish a lower bound on the score, we appeal to \eqref{equ:IP-to-be-lbed}, which tells us that
\begin{align}
\E_{\theta\sim V}\left( \E_{v,\calA(v)}\left[ \sum_{j=1}^n \langle \calA(v), v^j - \mu_\theta \rangle \right] \right) = \frac{\Area(S)}{\Vol(V)} \E_{\theta\sim S}[\langle g(\theta), \vec{n}_\theta\rangle]. \label{equ:hypercube-divergence-applied}
\end{align}
If $S$ is a sphere, we have that $\vec{n}_\theta = \frac{\theta}{\|\theta\|_2}$. Assuming $\calA$ is accurate, we have
$$
\| g(\theta) - \mu_\theta \|_2 \le \frac{1}{100}\sqrt{d}.
$$
Hence,
$$
\E_{\theta\sim S}[\langle g(\theta), \vec{n}_\theta\rangle] \ge \E_{\theta}[\langle \mu_\theta, \vec{n}_\theta\rangle] - \frac{1}{100}\sqrt{d}.
$$
Let us study the term $\langle \mu_\theta, \vec{n}_\theta\rangle$ closely. By definition of $\mu_\theta$, we have
$$
\langle \mu_\theta, \vec{n}_\theta\rangle = \sum_{i=1}^{d} \E_{v'\sim D_\theta}\left[v'_i\cdot \frac{\theta_i}{\|\theta\|_2}\right].
$$
Recall that we have $\Pr[v'_i = +1] = \frac{\exp(\theta_i)}{\exp(\theta_i)+\exp(-\theta_i)}$, meaning that the sign of $v'_i$ agrees with $\theta_i$ more often than not, which means that $\E_{v'\sim D_\theta}\left[v'_i\cdot \frac{\theta_i}{\|\theta\|_2}\right] \ge 0$ is always true. Next, for every $\theta_i$ with $|\theta_i| > 1$, we have $\Pr[v'_i = \sign(\theta_i)] > \frac{e}{e+e^{-1}} > \frac{2}{3}$, which implies $\E[v'_i \cdot \frac{\theta_i}{\|\theta\|_2}] > \Omega\left( \frac{1}{\|\theta\|_2}\right)$. Lastly, recall that $\theta$ is a random point on a sphere of radius $R = 5\sqrt{d}$. Hence, on average, there will be $\Omega(d)$ coordinates $i$ with $|\theta_i| > 1$. All in all, we conclude that
$$
\sum_{i=1}^{d} \E_{v'\sim D_\theta}\left[v'_i\cdot \frac{\theta_i}{\|\theta_i\|}\right] > \Omega(d)\cdot \Omega\left( \frac{1}{\|\theta\|_2 } \right) \ge \Omega\left( \frac{d}{R} \right).
$$
Hence, we obtain
\begin{align}
\eqref{equ:hypercube-divergence-applied}
&= \frac{\Area(S)}{\Vol(V)} \cdot \left( \E_\theta[\langle \mu_\theta, \vec{n}_\theta\rangle] - \frac{1}{100}\sqrt{d} \right).\notag \\
&\ge \frac{d}{R} \cdot \Omega\left(\frac{d}{R} \right) \notag \\
& \ge \Omega(d). \label{equ:hypercube-lb-finish}
\end{align}

\paragraph*{Wrap-up.} We now compare \eqref{equ:hypercube-lb-finish} with \eqref{equ:hypercube-ub-finish}. Observe that \eqref{equ:hypercube-lb-finish} says that the total score over all inputs must be at least $\Omega(d)$ for the algorithm to be accurate, while \eqref{equ:hypercube-ub-finish} says that each input contributes at most $O(\eps\sqrt{d})$ to the score (assuming $\delta$ is negligible). This immediately yields that there has to be at least $n\ge \Omega\left(\frac{\sqrt{d}}{\eps}\right)$ samples as desired.

\section{Composition of Query Lower Bounds}\label{sec:composition}

In this section, we extend the framework introduced in Section~\ref{sec:fingerprinting-meets-geometry} to prove our results on the ``composition'' of query releasing lower bounds. Compared with \cite{BunUV14}, the main innovation in our proof is to write the whole lower bound as a correlation analysis through the geometric fingerprinting framework. This section also serves as a warm-up to later applications of adaptive data analysis.

The main results covered in this section are as follows.

\begin{theorem}\label{theo:composition-lower-bound}
    Let $M,N$ be such that $M\gg \log(N)$ and $N\gg \log(M)$. There is a workload matrix $A\in \{\pm 1\}^{M\times N}$ such that, for every $\eps\in (0,1), \alpha \in (2^{-\log^{1/9}(N)}, 1)$, $\delta \in (\frac{1}{N}, \frac{1}{\log(N)\log(M)})$, any $(\eps,\delta)$-DP query releasing algorithm for $A$ requires $$
    n\ge \Omega\left( \frac{\sqrt{\log(N) \log(1/\delta)} \log(M)}{\eps\alpha^2}\right)
    $$
    samples to answer all $M$ queries within $\ell_{\infty}$-error at most $\alpha$.
\end{theorem}

We define two-way marginal queries. Construct a matrix $A_{TW}\in \{\pm 1\}^{d^2\times 2^{2d}}$ where the rows are indexed by pairs $(i,j)\in [d\times d]$ and columns indexed by $(x,y)\in \{\pm 1\}^{d+d}$. Then, set $(A_{TW})_{(i,j),(x,y)} = x_i\cdot y_j$. This is slightly different than the standard two-way marginals: here, every data point $(x,y)$ has $2d$ attributes, but we only query all the correlations crossing $x$ and $y$. However, since these queries are roughly half of all two-way marginal queries over $(x,y)$, lower bounds against $A_{TW}$ lift to lower bounds against the standard two-way marginal query family up to a constant factor.

\begin{theorem}\label{theo:two-way-marginal-lb}
    Let $A_{TW}\in \{\pm 1\}^{d^2\times 2^{2d}}$ be the two-way marginal query matrix. For every $\eps\in (0,1), \alpha \in (d^{-0.49},1/100)$, $\delta \in (2^{-\frac{d}{\log^2(d)}}, \frac{1}{d^2})$, any $(\eps,\delta)$-DP query releasing algorithm for $A$ requires
    $$
    n\ge \Omega\left( \frac{\sqrt{d\log(1/\delta)}}{\eps \alpha^2} \right)
    $$
    samples to achieve $\ell_2^2$-error of $\alpha^2 d^2$.
\end{theorem}

The proofs of both theorems are largely similar, with Theorem~\ref{theo:composition-lower-bound} being slightly more complicated. We prove Theorem~\ref{theo:composition-lower-bound} in Sections~\ref{sec:composition-setup} to \ref{sec:composition-finish}. Next, we explain the necessary modifications to prove Theorem~\ref{theo:two-way-marginal-lb} in Section~\ref{sec:composition-two-way}.

\subsection{Setup for the Fingerprinting Argument}\label{sec:composition-setup}

Given $N,M,\alpha, \eps,\delta$ as in Theorem~\ref{theo:composition-lower-bound}, we choose some $m = \Theta(\log(M))$ and $d=\Theta(\log(N))$ (constants to be specified). We also set $k = \frac{c}{\alpha^2}$ for some small constant $c > 0$. Let $\{e^i\}_{i=1}^{m}\subseteq \mathbb{R}^m$ be the standard basis and $\{u^j\}_{j=1}^{k}\subseteq \mathbb{R}^k$ be a collection of $k$ pairwise orthogonal Boolean vectors (e.g., Hadamard basis properly scaled). Consider now the following ensemble of vectors:
$$
K = \left\{ e^i \otimes u^j \otimes v: i\in [m], j\in [k],  v\in \{\pm 1\}^d \right\}.
$$
We choose $m,d$ properly so that $N = 2^d\times k\times m$ and $M = 2^{m}\times d\times k$.

\paragraph*{Construction of the query matrix.} We construct the query matrix $A\in \{\pm 1\}^{M\times N}$ now.
\begin{itemize}
\item Every vector $e^i\otimes u^j\otimes v$ identifies a column of $A$.
\item The rows of $A$ are indexed by tuples $(h,p,q)$ where $h:[m]\to \{\pm 1\}$ is a Boolean predicate, $p\in [k]$ and $q\in [d]$ are two indices. Note that there are $2^m\times d\times k$ such tuples.
\item Finally, for an entry of $A$ indexed by row $(h, p, q)$ and column $(e^i\otimes u^j\otimes v)$, set the entry to be $h(i)\cdot u^j_p \cdot v_q$.
\end{itemize}
We aim to prove Theorem~\ref{theo:composition-lower-bound} for the constructed $A$, via the geometric fingerprinting framework.

\paragraph*{Setup of the fingerprinting argument.} We define a family of distributions, which is slightly different than the standard exponential family. In particular, we choose the space of $\theta$ to be
$$
V = \left\{\theta\in \mathbb{R}^{m\times k\times d}: \sum_{i} |\theta_i|\le R \right\}.
$$
We will choose $R$ to be $\frac{mkd}{\sqrt{k}}$. In this way, typically every entry of $\theta$ is roughly $\frac{1}{\sqrt{k}}$. For every $\theta$, we define a \emph{type-conditioned} exponential distribution $D_{\theta}$ as follows:

\begin{itemize}
\item To sample from $D_{\theta}$, first choose a random $e^i$ and a random $u^j$.
\item Then select $v\in \{\pm 1\}^d$ with probability proportional to $\exp(\langle \theta, e^i\otimes u^j\otimes v\rangle)$.
\end{itemize}

In defining $D_{\theta}$, we insist that each $e^i$ and $u^j$ have equal probability of being chosen. For this reason, the standard divergence-to-score lemma (Proposition~\ref{prop:differentiate-to-score}) does not apply as is. Still, for every $x\in \mathrm{supp}(D_{\theta})$, define the \emph{type} of $x$ as the pair $(i,j)$ if $x = e^i\otimes u^j\otimes v$ for some $v$. Then, let $\calA$ be an algorithm operating on $n$ iid samples from $D_{\theta}$ and outputting an $(mkd)$-dimensional vector. We can prove the following proposition.

\begin{proposition}\label{prop:type-dependent-score-1}
For every $i\in [mkd]$, it holds that
$$
\frac{\partial \E_{x\sim D_{\theta}^n}[\calA(x)_i]}{\partial \theta_i} = \sum_{j=1}^n \E_{x\sim D_{\theta}^n}\left[ \calA(x)_i \cdot \left(x^j_i - \E_{x'\sim D_{\theta} \mid type(x') = type(x^j)}[x'_i] \right) \right].
$$
Consequently, it holds that
$$
\mathrm{div}~\mathbb{E}_{x\sim D_{\theta}^{n}}[\calA(x)] = \mathbb{E}\left[\sum_{j\in [n]}\langle \calA(x), x^j - \mathbb{E}_{x':type(x') = type(x^j)}[x']\rangle\right]
$$
\end{proposition}

\begin{proof}[Proof sketch.]
The proof is nearly identical to that of Proposition~\ref{prop:differentiate-to-score}, except for one detail: in the new setup, for every $x = e^i\otimes u^j\otimes v$, we have
$$
\begin{aligned}
\frac{\partial \Pr[D_\theta = x]}{\partial{\theta_p}}
&= \frac{\partial}{\partial \theta_p} \left( \frac{1}{mk} \cdot \frac{\exp(\langle e^i\otimes u^j\otimes v,\theta\rangle)}{\sum_{v'} \exp(\langle e^i\otimes u^j\otimes v,\theta\rangle)}\right) \\
&= \frac{1}{mk} \left( \frac{\exp(\langle e^i\otimes u^j\otimes v,\theta\rangle) \cdot (e^i \otimes u^j\otimes v)_p }{\sum_{v'} \exp(\langle e^i\otimes u^j\otimes v,\theta\rangle)} \right. -\\
&~~~~~~~~ \left. \frac{\exp(\langle e^i\otimes u^j\otimes v,\theta\rangle) \sum_{v'} \exp(\langle e^i\otimes u^j\otimes v', \theta\rangle) \cdot (e^i \otimes u^j\otimes v')_p }{ \left( \sum_{v'} \exp(\langle e^i\otimes u^j\otimes v,\theta\rangle) \right)^2}\right) \\
&= \Pr[D_\theta=x] \cdot \left( x_p - \E_{x':type(x') = type(x)}[ x'_p ] \right).
\end{aligned}
$$
To prove Proposition~\ref{prop:type-dependent-score-1}, we repeat the lines of the proof for Proposition~\ref{prop:differentiate-to-score} and use the new derivative formula to replace the appearance of ``$\frac{\partial \Pr[D_\theta = x]}{\partial \theta_i}$'' there as appropriate.
\end{proof}

Proposition~\ref{prop:type-dependent-score-1} suggests defining the score of data points in a \emph{type-dependent} way: given a vector $q\in \mathbb{R}^{m\times k\times d}$ and the distribution $D_{\theta}$, for every $x\in \mathrm{supp}(D_\theta)$, we define the score of $x$ with respect to $q$ as
$$
score(x) \coloneqq \mathbb{E}_{x'\sim D_{\theta}:type(x')=type(x)}[\langle q, x - x'\rangle].
$$
This new score definition enjoys all the properties we need: the score of a fresh sample is zero on average and well concentrated, while the average score of in-sample data points is significantly larger, as we will show in a moment.

\subsection{``Surgery'' on the Query-Releasing Algorithm} \label{sec:composition-surgery}

To prove Theorem~\ref{theo:composition-lower-bound}, we assume there is a sample-efficient private algorithm $\calA$ and try to derive a contradiction. However, the output of $\calA$ is an $M$-dimensional vector (since there are $M$ queries in total), while the fingerprinting argument requires the dimension of $\theta$ and that of the algorithm to match. To reconcile the mismatch, we post-process the output of $\calA$, and get an estimate to the mean of $D_{\theta}$.

In this subsection, we interpret $K = \mathrm{supp}(D_{\theta})$ as both a discrete set with support size $mk2^{d}$ (the discrete perspective) and an ensemble of vectors living in $\mathbb{R}^{m\times k\times d}$ (the geometric perspective).

\paragraph*{Understanding the output of $\calA$.} Taking the first interpretation, by inspecting the design of $A\in \{\pm 1\}^{M\times N}$, we see that $\calA$ is answering the following ensemble of statistical queries:
$$
\left\{ h(e^i\otimes u^j\otimes v) \coloneqq g(i)\cdot (u^j)_\ell \cdot v_r \mid g:[m]\to \{\pm 1\}, \ell\in [k], r\in [d]  \right\}.
$$
Note that there are $2^{m}\cdot k\cdot d$ queries in total. Say we have received $\alpha$-accurate answers to these queries. Let $\{\hat{a}_{g,\ell, r}\}$ be the collection of answers.

Now, we switch from the discrete perspective on $\mathrm{supp}(D_{\theta})$ to a geometric perspective. Before we start to post-process $\{\hat{a}_{g,\ell,r}\}$, we denote $\mu_{\theta}=\mathbb{E}[D_{\theta}]$ to be the mean of $D_{\theta}$ and observe the following useful properties about it:

\begin{itemize}
\item For every $i\in [m], \ell\in[k], r\in[d]$, we have $|(\mu_{\theta})_{i,j,k}| \in 0 \pm \frac{1}{m}$. This is simply because each $e_i$ is sampled by $D_{\theta}$ with probability $\frac{1}{m}$.
\item For every $i\in [m]$ and $r\in [d]$, we have $|\langle (\mu_{\theta})_{i,*,r}, u^j\rangle| \le \frac{1}{m}\cdot \frac{\langle u^j, u^j\rangle}{k} \le \frac{1}{m}$. This follows because $\{ u^j \}$ are mutually orthogonal, and each $u^j$ gets sampled with probability $\frac{1}{k}$.
\end{itemize}

We apply the following two post-processing on $\{\hat{a}_{g,\ell, r}\}$ to find a ``nice'' $\hat{\mu}$ that approximates $\mu_{\theta}$ well.

\paragraph*{The reconstruction argument.} First, for every $\ell\in [k]$ and $r\in [d]$, we can find a vector $\tilde{\mu}_{*,\ell, r}\in [\pm 1/m]^{m}$ such that
$$
\|\tilde{\mu}_{*,\ell, r} - (\mu_{\theta})_{*,\ell, r}\|_1 \le 2\alpha.
$$
To achieve this, note that for every $g:[m]\to \{\pm 1\}$, the output $\hat{a}_{g,\ell,r}$ places an inner-product restriction of $\mu_{\theta}$ by requiring that $\sum_{i=1}^m (\mu_\theta)_{i,\ell,r} \cdot g(i) \in \hat{a}_{g,\ell,r} \pm \alpha$. In light of this observation, all we need to do is to find a vector $\tilde{\mu}_{*,\ell, r}$ that is consistent with all the restrictions from $\{\hat{a}_{*,\ell,r}\}_{g}$ up to additive error $\alpha$. Such a vector exists because $(\mu_{\theta})_{*,\ell, r}$ is one. Moreover, since every statistical test $g:[m]\to \{\pm 1\}$ cannot tell $(\mu_{\theta})_{*,\ell, r}$ and $\tilde{\mu}_{*,\ell,r}$ apart, we know the two vectors are close in $\ell_1$ distance\footnote{We note that this is essentially the reconstruction argument for query classes with large VC dimension: c.f.~\cite[Section 5.1.1]{BunUV14}}.

\paragraph*{A projection step.} To bound the ``variance'' of the score later, we want to apply one more surgery on $\tilde{\mu}$. Specifically, for every $i\in [m]$ and $r\in [d]$, we consider the vector $\tilde{\mu}_{i,*,r}$ and project it to the following space by minimizing $\ell_1$-movement:
$$
H\coloneqq \left\{\frac{1}{km}\sum_{j=1}^{k} \lambda_j u^j:\lambda \in [-1,1]^k\right\}.
$$
Denote the resulting vector to be $\hat{\mu}_{i,*,r}$. To analyze the error, note that $(\mu_{\theta})_{i,*,r}$ lies in the set $H$ and it is close to $\tilde{\mu}_{i,*,r}$ in $\ell_1$-distance. Hence, an application of triangle inequality shows that
$$
\|\hat{\mu} - \mu_{\theta}\|_1 \le 2 \|\tilde{\mu} - \mu_{\theta}\|_1 \le 4\alpha kd.
$$

We remark here that this argument uses the assumption that the algorithm $\calA$ is accurate for the population, rather than just for the sample.

\paragraph*{Summary.} We compose the post-processing procedure above with the algorithm $\calA$. This gives us an algorithm $\calB$ which, on input $n$ iid samples from $D_{\theta}$, returns an approximate average $\hat{u}\in \mathbb{R}^{m\times k\times d}$ with the following guarantees:

\begin{itemize}
\item For every $i\in [m]$, every $r\in [d]$ and $u^j$, we have $|\langle \hat{\mu}_{i,*,r}, u^j\rangle|\le \frac{1}{m}$.
\item Letting $\mu_\theta =\mathbb{E}[D_\theta]$ be the true mean, we have $\| \mu_\theta - \hat{\mu} \|_1\le 4\alpha k d$.
\end{itemize}

For convenience, we define $\hat{\mu}(\theta) \coloneqq \E_{x\sim D_{\theta}^n}[\calB(x)]$ to be the average output of $\calB$ on a data drawn from $D_{\theta}$. We understand $\hat{\mu}$ as a vector field over $\mathbb{R}^{mkd}$.

\subsection{Proof of Theorem~\ref{theo:composition-lower-bound}}\label{sec:composition-finish}

We are ready to run the fingerprinting argument and finish the proof of Theorem~\ref{theo:composition-lower-bound}.

\paragraph*{Score Lower Bound.} First, let us lower bound the surface integral and, consequently, the average divergence. Recall the space of $\theta$ is given by
$$
V = \left\{ \theta\in \mathbb{R}^{m\times k\times d}: \|\theta\|_1 \le \frac{mkd}{\sqrt{k}} \right\}.
$$
Let $S = \partial V$ be the surface of $V$. We claim the following:

\begin{itemize}
\item For $V$ being a $d$-dimensional $\ell_1$-ball of radius $R$ and $S$ its boundary, one has $\frac{\Area(S)}{\Vol(V)} = \frac{d\sqrt{d}}{R}$.
\item Back to our example, $V$ is a $(mkd)$-dimensional $\ell_1$-ball of radius $\frac{mkd}{\sqrt{k}}$. Hence, it holds that $\frac{\Area(S)}{\Vol(V)} = \frac{mkd\sqrt{mkd}\cdot \sqrt{k}}{mkd}=k\sqrt{md}$.
\item For every $\theta\in S$, the unit normal vector at $\theta$ against $S$ is $\frac{1}{\sqrt{mkd}}\sign(\theta)$.
\end{itemize}

\newcommand{\score}{\mathrm{score}}

Fix a $\theta$ and consequently $D_{\theta}$. For every $x\in K$ and $q\in \mathbb{R}^{mdk}$, denote
$$
\score(x;q) \coloneqq \left\langle x - \E_{x'\sim D_{\theta}\mid type(x) = type(x')}[x'], q \right\rangle.
$$
Using Proposition~\ref{prop:type-dependent-score-1} and the divergence theorem, we obtain: 
\begin{align}
&~~~~ \E_{\theta\sim V} \E_{x^1,\dots, x^n\sim D_{\theta}, \calB(x)} \left[ \sum_{j=1}^{n} \score(x^j, \calB(x)) \right] \notag \\
& = \E_{\theta\sim V} \left[ \diver \E_{x\sim D_\theta^n,\calB(x)}[\calB(x)] \right] \notag \\
&= \frac{\Area(S)}{\Vol(V)} \E_{\theta\sim S}\left[ \left\langle \frac{\sign(\theta)}{\sqrt{mkd}}, \E_{x}[\calB(x)] \right\rangle \right] \notag \\
&\ge \frac{\Area(S)}{\Vol(V)} \left( \E_{\theta\sim S}\left[ \left\langle \frac{\sign(\theta)}{\sqrt{mkd}}, \mu_{\theta} \right\rangle \right] - \frac{4\alpha kd}{\sqrt{mkd}} \right). \label{equ:composition-surface-integral}
\end{align}
Here, the last inequality is true because we have shown that $\| \hat{\mu}(\theta) - \mu(\theta) \|_{1} \le 4\alpha kd$ in Section~\ref{sec:composition-surgery}.

Now let us study the term $\E_{\theta\sim S}\left[\langle \sign(\theta), \E[D_{\theta}]\rangle\right]$. By definition of $D_{\theta}$, we can first sample and condition on $e^i$ and $u^j$. Then, we would like to understand the average of the following inner product (the ``dot'' below denotes tensor multiplication):
\begin{align}
\E_{\theta} \left\langle \sign(\theta)\cdot (e^i \otimes u^j), \E_{(e^i\otimes u^j\otimes v)\sim D_{\theta} \mid e^i\otimes u^j }[ v ] \right\rangle. \label{equ:IP-after-conditioning-ei-ej}
\end{align}
Over a random $\theta\sim S$, we see that $\sign(\theta) \cdot (e^i\otimes u^j)$ is a $d$-dimensional vector that is entry-wise $\Theta(\sqrt{k})$.\footnote{To see this, interpret $\sign(\theta)$ as stacking of $m$ matrices each of dimension $k\times d$. $\sign(\theta)\cdot e^i$ restricts onto one of the matrices, and $\sign(\theta)\cdot (e^i\otimes u^j)$ further takes a signed summation of the $k$ rows, where the signs are given by $u^j$. The claim that $\sign(\theta) \cdot (e^i\otimes u^j)$ is entry-wise $\Theta(\sqrt{k})$ follows from the randomness of $\theta$.} Also, recall that
$$
\Pr_{(e^i\otimes u^j\otimes v\sim D_{\theta})\mid e^i\otimes u^j} [v] \propto \exp(\langle \theta \cdot (e^i\otimes u^j), v\rangle).
$$
Over a random $\theta\sim S$, the vector $\theta \cdot (e^i\otimes u^j)$ is entrywise $\Theta(1)$, implying that each entry of $v$ has a constant bias toward $\sign(\theta \cdot (e^i\otimes u^j))$. Since $\sign(\theta \cdot (e^i\otimes u^j))$ and $\sign(\theta) \cdot (e^i\otimes u^j)$ are highly correlated, we conclude that
\begin{align}
\eqref{equ:IP-after-conditioning-ei-ej} \ge \Omega( d) \cdot \Omega(\sqrt{k}) \ge \Omega(d\sqrt{k}). \label{equ:composition-IP-lb}
\end{align}
Recall that we set $k = \frac{c}{\alpha^2}$ for a small constant $c$. Hence, from~\ref{equ:composition-surface-integral} and \eqref{equ:composition-IP-lb}, we see that
$$
\eqref{equ:composition-surface-integral} \ge \frac{\Area(S)}{\Vol(V)} \left( \frac{d\sqrt{k}}{\sqrt{mkd}} - \frac{4\sqrt{c}\cdot \sqrt{k}d}{\sqrt{mkd}} \right) \ge \Omega\left( k\sqrt{md} \cdot \frac{d\sqrt{k}}{\sqrt{mkd}} \right) \ge \Omega\left( kd \right).
$$

\paragraph*{Score upper bound.} We now derive an upper bound on the score. We could easily repeat the argument in Section~\ref{sec:fingerprinting-hypercube}, but then we would end up with a lower bound of the same order as \cite{BunUV14} without achieving the additional $\sqrt{\log(1/\delta)}$ factor.

To unleash the full power of our framework, we borrow a trick from \cite{SteinkeU16}. In particular, we use the following connection established via group privacy.
\begin{lemma}[\cite{SteinkeU16}]\label{lemma:group-privacy-reduce-sample}
    Let $\eps \in (0,1)$ and $\delta < 1$. Suppose there is a $(\eps,\delta)$-DP algorithm for private query releasing on $A\in \{\pm 1\}^{N\times M}$ with sample complexity $n$. Then, for every $p\ge 1$, there is a $(p\eps, \frac{e^{p\eps}-1}{e^\eps - 1}\delta)$-DP algorithm for the same task with sample complexity $\frac{n}{p}$.
\end{lemma}

Back to our example, we pick $p = \frac{\log(1/\delta)}{4 \eps}$. In order to prove Theorem~\ref{theo:composition-lower-bound}, by Lemma~\ref{lemma:group-privacy-reduce-sample}, it suffices to show that there is \emph{no} $(\log(1/\delta),\sqrt{\delta})$-DP algorithm with sample complexity $n\le o\left(  \frac{\sqrt{\log(N)} \log(M)}{\sqrt{\log(1/\delta)} \alpha^2} \right)$, as this would ``lift'' to a $\Omega(\frac{\sqrt{\log(N) \log(1/\delta)}\log(M)}{\eps\alpha^2})$ lower bound for the original parameter setting.

Suppose for contradiction that such an algorithm $\calB$ exists. The lower bound part tells us the total score must be $\Omega(kd)$. To establish the score upper bound, we need a stronger proposition than Proposition~\ref{prop:privacy-imply-small-score-difference}, utilizing the concentration property on the score of an independent point. Namely, we claim:
\begin{proposition}\label{prop:privacy-to-small-score-subgaussian}
    Suppose $X$ and $X'$ are a pair of $(\eps,\delta)$-indistinguishable random variables supported on $[-B, B]$, such that $\E[X] = 0$ and $\Pr[X > t] \le e^{-\frac{t^2}{C^2}}$ for every $t\ge 4 C\sqrt{\eps}$. Then,
    $$
    \E[X'] \le O(\delta B + C \sqrt{\eps}).
    $$
\end{proposition}
\begin{proof}
    Again, we pay a price of $B\delta$ to consider a random variable $X'$ that is $(\eps,0)$-indistinguishable from $X$. For this new $X'$, we have
    $$
    \begin{aligned}
    \E[X']
    &\le \int_{0}^B \Pr[X'\ge t] dt \\
    &\le 4 C\sqrt{\eps} + \int_{4 C\sqrt{\eps}}^{B} \Pr[X'\ge t] dt \\
    &\le 4 C\sqrt{\eps} + \int_{4 C\sqrt{\eps}}^{B} \exp(-\frac{t^2}{C^2} + \eps) dt \\
    &\le O(C\sqrt{\eps}).
    \end{aligned}
    $$
    This completes the proof.
\end{proof}

Consider now the random variable $\score(x';\calB(x))$, where the randomness is over $x',x$ and $\calB(x)$. We first condition on $x,\calB(x)$ and $type(x')$. Suppose $type(x') = e^i\otimes u^j$. It remains to sample the ``$v$'' part and calculate the score of $x' = e^i\otimes u^j\otimes v$. After conditioning on $e^i\otimes u^j$, the score can be equivalently written as
$$
\langle v - \E[v'], \calB\cdot (e^i\otimes u^j)\rangle.
$$
As $v$ is sampled from a Boolean cube according to an exponential tilt, it is clearly seen that each coordinate of $v$ contributes independently to the score. By the argument of Section~\ref{sec:composition-surgery}, each coordinate of $\calB\cdot (e^i\otimes u^j)$ is bounded by $\frac{1}{m}$. Therefore, we conclude that the score of a fresh point is $T$-subgaussian with $T = \Theta(\frac{\sqrt{d}}{m})$. Using Proposition~\ref{prop:privacy-to-small-score-subgaussian}, this implies that, for every in-sample data point $x^j$, it holds that
$$
\E_{x,\calB(x)}[\score(x^j;\calB(x)] \le O(\frac{\delta d}{m} + \frac{\sqrt{\log(1/\delta)} \sqrt{d}}{m}) \le O(\frac{\sqrt{d \log(1/\delta)}}{m}).
$$
But remember, the total score must be at least $\Omega(kd)$ for the algorithm to be accurate. As such, we conclude that the number of samples $\cal$ uses is at least
$$
\Omega\left(  \left(\frac{\sqrt{d \log(1/\delta)}}{m} \right)^{-1} \cdot kd  \right) \ge \Omega\left( \frac{\sqrt{d} m k}{\sqrt{\log(1/\delta)}} \right),
$$
which is what we desired, because we had set $d = \Theta(\log(N))$, $m = \Theta(\log(M))$ and $k = \Theta(\frac{1}{\alpha^2})$.

\begin{remark}\label{remark:low-sample-no-generalization}
    One might wonder whether one can take $\delta$ to be $2^{-d/\alpha^2}$, and prove a lower bound of $\frac{\sqrt{\log N}\log M}{\alpha^3 \eps}$ (this would imply a new lower bound against pure-DP algorithms). This is not possible under the reduction of Lemma~\ref{lemma:group-privacy-reduce-sample}. Because by taking $p = \frac{\log(1/\delta)}{\eps}$ there, the sample complexity lower bound we aim for is of order $\frac{\log M}{\alpha}$. However, this number is so small that a sample-accurate algorithm does not necessarily generalize, and the projection step in our argument would fail.%fingerprinting framework does not apply.\footnote{Put in other words, the fingerprinting proof does rule out all algorithms that are accurate w.r.t.~the distribution. However, with $\frac{\log(M)}{\alpha}$ samples, a sample-accurate algorithm needs not to generalize. Hence, the proof says nothing about sample-accurate algorithms.} It turns out 
    The lowest $\delta$ we can ``afford'' is $\delta \approx 2^{-d}$, which asks us to prove a lower bound of order $\frac{\log(M)}{\alpha^2}$. In this regime, any sample-accurate algorithm still generalizes, and our proof technique applies.
\end{remark}

\subsection{Lower Bounds for Two-Way Marginals}\label{sec:composition-two-way}

In this section, we prove Theorem~\ref{theo:two-way-marginal-lb}. The proof structure is largely similar to that of Theorem~\ref{theo:composition-lower-bound}, with a couple of key differences we highlight below:
\begin{itemize}
    \item Theorem~\ref{theo:composition-lower-bound} requires composing three lower bounds (the $\Omega(m)$ and $\Omega(1/\alpha^2)$ reconstruction lower bounds and the $\Omega(\sqrt{d})$ one-way marginal lower bounds). For two-way marginals, we only compose an $\Omega(1/\alpha^2)$ reconstruction lower bound with a $\Omega(\sqrt{d})$ one-way marginal lower bound.
    \item Since we are interested in $\ell_2^2$ error, we will set the space of $\theta$ be an $\ell_2$-ball of appropriate radius and carry out the analysis.
\end{itemize}

We detail the argument below.

\subsubsection{Setup for Fingerprinting Argument}

Recall that $\alpha \gg \frac{1}{\sqrt{d}}$. We set $k = \frac{c}{\alpha^2}$ for some small $c > 0$. Let $\{u^j\}_{j=1}^{m}\subseteq \mathbb{R}^k$ be a collection of $k$ pairwise orthogonal Boolean vectors. We work with the following ensemble of vectors:
$$
K = \{u^j\otimes v : v\in \{\pm 1\}^d\}.
$$
We define a query matrix $A:\{\pm 1\}^{kd\times k 2^d}$ from $K$, by naturally concatenating all vectors in $K$ as column vectors. When $k\approx d$, the matrix $A$ is a sub-matrix of the two-way marginal query matrix $A_{TW} \in \{\pm 1\}^{d^2\times (2^{2d})}$. Namely, $A$ restricts to a subset of all possible attribute profiles. For this reason, lower bounds against $A$ lift to that against $A_{TW}$ naturally.

For the case $k < d$, we may duplicate each entry of $u^j$ by $\frac{d}{k}$ times, and get a query matrix $A'\in \{\pm 1\}^{d^2\times (k\times 2^d)}$ that is a sub-matrix of $A_{TW}$. This does not significantly change the privacy-utility trade-off: namely, a ($\ell_2^2$) error-sample lower bound of $(\alpha^2 dk, n)$ for $A$ translates directly to a lower bound of $(\alpha^2 d^2,n)$ for $A'$, which would imply the same lower bound for $A_{TW}$. However, for us it will be more convenient to work with the $kd$-by-$k2^d$ matrix.

Continuing, we choose the space of $\theta$ as $V = B_{\ell_2}(0, R) 
\subseteq \mathbb{R}^{kd}$ with radius $R = \sqrt{d}$. Note that for a typical $\theta\sim V$, each entry of $\theta$ is roughly $\frac{1}{\sqrt{k}}$. For every $\theta\in V$, again define the type-conditioned exponential tilt $D_{\theta}$ as
\begin{itemize}
    \item To sample from $D_{\theta}$, first select a random $u^j$.
    \item Then, select $v\in \{\pm 1\}^d$ with probability proportional to $\exp(\langle \theta, u^j\otimes v\rangle)$.
\end{itemize}
A version of Proposition~\ref{prop:type-dependent-score-1} holds in this case. Namely, define the type of a vector $x=u^j\otimes v$ as $u^j$. We define the score of $x$ w.r.t.~$q\in \mathbb{R}^{kd}$ as $\score(x;q) \coloneqq \left\langle x - \E_{x':type(x')=type(x)}[x'], q \right\rangle$. Similarly to Proposition~\ref{prop:type-dependent-score-1}, we have
$$
\diver \E_{x\sim D_\theta^n}[\calA(x)] = \E\left[ \sum_{j\in [n]}\langle \calA(x), x^j - \E_{x':type(x')=type(x^j)}[x]  \rangle\right].
$$
We define the type-conditioned score function $\score(x';\calA(x))$ accordingly.

\subsubsection{Post-Processing the Algorithm}\label{sec:composition-two-way-surgery}

Suppose there is an algorithm $\calA$ for query release with workload matrix $A$, and $\calA$ has $\ell_2^2$-error at most $\alpha^2 kd$. Let $\tilde{\mu}$ be the output of $\calA(x)$ where $x\sim D_{\theta}^n$. Let $\mu_\theta$ denote $\E[D_{\theta}]$. We know that $\|\tilde{\mu} - \mu_\theta\|_2 \le \alpha \sqrt{kd}$.  Furthermore, we have that for every $r\in [d]$, the vector $(\mu_{\theta})_{*,r}\in\mathbb{R}^k$ satisfies that $\langle (\mu_\theta)_{*,r},u^j\rangle \le \frac{1}{k}\langle u^j, u^j\rangle \le 1$.

Similarly as in Section~\ref{sec:composition-surgery}, for every $r\in [d]$, we project $\tilde{\mu}_{*,r}$ to the set
$$
H = \left \{ \frac{1}{k} \sum_{j=1}^{k} \lambda_j u^j :  \lambda \in [-1,1]^k  \right\}
$$
Denote the result to be $\hat{\mu}$. Note that $(\mu_\theta)_{*,r}$ is in the set $H$ and it is close to $\tilde{\mu}$ in $\ell_2$ distance. An application of triangle inequality shows that $\|\hat{\mu} - \mu_\theta\|_2 \le 2\alpha\sqrt{kd}$.

\subsubsection{Correlation Analysis}

Now, let us give contradicting upper and lower bounds on the score of the algorithm.

\paragraph*{Score lower bound.} In terms of lower bound, we have
\begin{align}
\E_{\theta\sim V} \E_{x\sim D_\theta^n,\calA(x)} \left[\sum_{j=1}^{n} \score(x^j;\calA(x)) \right]
&= \frac{\Area(S)}{\Vol(V)} \E_{\theta\sim S} \left[\langle \E[\calA(x)], \frac{\theta}{\|\theta\|_2} \rangle \right] \notag \\
&\ge \frac{\Area(S)}{\Vol(V)} \E_{\theta\sim S} \left[\left\langle \E[D_\theta], \frac{\theta}{\|\theta\|_2} \right\rangle - \alpha \sqrt{kd} \right]. \label{equ:two-way-score-lb}
\end{align}
We use a similar argument as in Section~\ref{sec:composition-finish}: to understand $\langle \E[D_\theta], \theta\rangle$, we first sample and condition on $u^j$ (independent of $\theta$). Then, we would like to lower bound
$$
\E_{\theta} [\langle \E_{u^j\otimes v\sim D_{\theta}\mid u^j}[v], \theta \cdot u^j \rangle].
$$
Recall that $\theta \sim \partial B_2(0,\sqrt{d})$. Hence, each entry of $\theta$ is typically $\Theta(1/\sqrt{k})$. Consequently, the vector $\theta u^j$ is typically entry-wise $\Theta(1)$. Hence, the inner product above is $\Omega(d)$. Since $\|\theta\|_2 = \sqrt{d}$, we conclude that
$$
\eqref{equ:two-way-score-lb} \ge \frac{\Area(S)}{\Vol(V)}[ \Omega(\sqrt{d}) - \alpha \sqrt{kd} ] \ge \Omega(\frac{kd}{\sqrt{d}}) \cdot \Omega(\sqrt{d}) \ge \Omega(k d).
$$
Here, the second inequality is valid so long as we choose $k = \frac{c}{\alpha^2}$ with a sufficiently small $c > 0$, so that $\alpha\sqrt{kd}$ is dominated by the first term $\Omega(\sqrt{d})$.

\paragraph*{Score upper bound.} The rest of the argument is largely similar to that in Section~\ref{sec:composition-finish}. Condition on the output of the algorithm $\calA(x)$, we argue that the score of a fresh point $x'\sim D_\theta$ is well concentrated. Note that, conditioning on the type of $x'$ being $u^j$, the score of the  point $x'=u^j\otimes v$ can be equivalently written as
$$
\langle v - \E_{v'}[v'], \calA(x)\cdot u^j\rangle.
$$
Since $v$ is sampled from the Boolean cube according to an exponential tilt, each coordinate of $v$ contributes independently to the score. Since each coordinate of $\calA(x)\cdot u^j$ is bounded by $1$ (see Section~\ref{sec:composition-two-way-surgery}), we conclude that the score of a random $v$ is $K$-subgaussian with $K = O(\sqrt{d})$.

Next, for any desired $\eps \in (0, 1)$ and $\delta \in (2^{-\frac{d}{\log^2(d)}}, 1/d^2)$, we can first prove that, any $(\log(1/\delta), \sqrt{\delta})$-DP algorithm has sample complexity lower bound of
$$
n\ge \Omega(\frac{kd}{\sqrt{d} \cdot \sqrt{\log(1/\delta)}}) \ge \Omega(\frac{\sqrt{d}}{\sqrt{\log(1/\delta)}\alpha^2}).
$$
We emphasize that, thanks to the choice of $\delta$, our target lower bound is above $\frac{\log(d)}{\alpha^2}$, which is the number of samples required for low generalization error. In this regime, the fingerprinting framework applies. However, we cannot set $\delta$ to be much lower: when $\delta < 2^{-\omega(d)}$, the target lower bound becomes $o(1/\alpha^2)$, and we can no longer derive a contradiction via the fingerprinting framework.

Lastly, using the group privacy connection of Lemma~\ref{lemma:group-privacy-reduce-sample}, we obtain a lower bound of
$$
n\ge \Omega(\frac{\sqrt{d \log(1/\delta)}}{\eps\alpha^2})
$$
for $(\eps,\delta)$-DP algorithms, proving Theorem~\ref{theo:two-way-marginal-lb}.

\section{Lower Bounds on Adaptive Data Analysis}\label{sec:lb-ada}

In this section, we ``lift'' the lower bounds proved in Section~\ref{sec:composition} to lower bounds on the task of adaptive data analysis.

\subsection{Overview and Intuition}

\paragraph*{Intuition.} Recall the proof of Theorem~\ref{theo:composition-lower-bound} (Section~\ref{sec:composition-finish}). There, we constructed a query matrix $A\in \{\pm 1\}^{M\times N}$, defined a family of distribution $\{D_\theta\}_{\theta\in V}$ and a type-dependent score accordingly. Then, we established the following: letting $\calA(x)$ be the algorithm's output on input $x = (x^1,\dots, x^n)$, the scores of in-sample data points w.r.t.~$\calA(x)$ behave like
$$
\sum_{j=1}^{n} \score(x^j;\calA(x)) \ge \Omega(dk) = \Omega(\frac{d}{\alpha^2}),
$$
while for an independent $x'\sim D_{\theta}$, with high probability we have $|\score(x';\calA(x))| \le \tilde{O}(\frac{\sqrt{d}}{m})$. Moreover, the average score of an independent $x'$ is zero.

Based on the argument so far, we would easily achieve an ADA lower bound if we could claim the following in addition:
\begin{description}
   \item[Even-ness] Every in-sample point $x^j$ has a score bounded by $\tilde{O}(\frac{\sqrt{d}}{m})$ with high probability, just like a fresh data point.
\end{description}

If this was indeed the case, we can craft a query $q$ by defining
$$
q(x) \coloneqq \max\left( -1, \min\left( 1, \score(x;\calA(x)) \cdot \frac{m}{100 \sqrt{d}} \right)\right).
$$
That is, we simply truncate the score of each $x\in [N]$ into the range $\pm O(\frac{\sqrt{d}}{m})$. Then, we scale down the score and use it to define a bounded linear query. If Even-ness holds, the truncation incurs little error to both the distribution and the data set. Hence, we end up with roughly $\E_{x\sim D_{\theta}}[q(x)]\approx 0$ and $\E_{j}[q(x^j)] \approx \Theta(\frac{1}{n}\cdot \frac{d}{\alpha^2}\cdot \frac{m}{\sqrt{d}}) = \Theta(\alpha)$ (recall that we assume $n\le \frac{c\sqrt{d}m}{\alpha^3}$ toward proving an ADA lower bound). We see that the query $q$ witnesses a mismatch between the distribution and the data set.

\paragraph*{Challenge and its resolution.} Unfortunately, Even-ness does not necessarily hold. As an example, if all the queries were given in a batch, the algorithm could simply use its first $\Theta(\frac{\log(M)}{\alpha^2})$ samples to evaluate the queries and ensure low generalization error. For such algorithms, the scores of the first few samples are huge (their privacy may be completely compromised), but the scores of other samples are small and have zero mean. We cannot exploit the tiny fraction of exposed data, as we must construct a linear query to witness a mismatch between the distribution and the whole data set.

To resolve the issue, we will design an adversary that sends the queries to $A$ in $d$ adaptive rounds. Meanwhile, we design a strategy to ``force'' the algorithm to use its samples evenly (and, consequently, distribute the score evenly).

At a very high level, we adopt a well-known trick, sometimes referred to as a ``one-time pad'' or a ``random mask''. In its most basic form, the idea is to generate a random Boolean string $r\sim \{\pm 1\}^{d}$ and sample a data point $x^j = (e^{i_j}\otimes u^{i'_j}\otimes v^j)$ as usual. However, instead of giving $x^j$ to the algorithm directly, we give algorithm the ``encrypted'' point $(e^{i_j}\otimes u^{i'_j}\otimes (r \text{ xor } v^j))$. As $r$ is random, the algorithm has no clue what the real $v^j$ is. Then, over the $d$ rounds of interaction, the adversary reveals $r$, and thus $v^j$, bit by bit. Throughout the process, the adversary keeps track of the score of $x^j$. Once it finds that the score is too large, it stops revealing future masks on $v^j$. In this way, the algorithm effectively loses access to $x^j$.

This oversimplified intuition hides many details, which can be found in the formal argument, to appear in the rest of the section.

\subsection{Partial Score and Fair Algorithms}

In this section, we consider the same setup as in Section~\ref{sec:composition-setup}, which we quickly review here. We have the set of vectors
$$
K = \left\{ e^i\otimes u^j \otimes v : i\in [m], j\in [k], v\in \{\pm 1\}^d \right\}.
$$
We choose $V = \left\{ \theta \in \mathbb{R}^{m\times k\times d} \right\}$, for every $\theta\in V$, we define the type-conditioned exponential distribution $D_{\theta}$.

The set $K$ induces a query matrix $A\in \{\pm 1\}^{M\times N}$ as in Section~\ref{sec:composition-setup}. Given an algorithm $\calA$ for answering statistical queries from $A$, we can apply the post-processing procedure of Section~\ref{sec:composition-surgery}. In the following, we will always work with mean-estimation algorithms $\calA:K^n\to \mathbb{R}^{mkd}$. It should be clear that such an algorithm is a direct product of any query-releasing algorithm via the reduction in Section~\ref{sec:composition-surgery}.

Given an algorithm $\calA:K^n\to \mathbb{R}^{mkd}$, we can run $\calA$ on an input $x\sim D_{\theta}^n$ and obtain $\calA(x)\in \mathbb{R}^{mkd}$. For every input point $x^j$, we define the score of $x^j$ w.r.t. $\mathcal{A}(x)$ as
$$
\score(x^j;\calA(x)) \coloneqq \left\langle x^j - \E_{x':type(x')=type(x^j)}[x'], \calA(x) \right\rangle.
$$
% We emphasize that the definition is not tied to the input/output of the algorithm. Namely, given $\theta\in V$, the quantity $\score(x;q)$ is well defined for every $x\in K$ and $q\in \mathbb{R}^{mkd}$.

\newcommand{\pscore}{\mathrm{pscore}}

\paragraph*{Slicing queries and Partial score.} Recall the query ensemble is $\{h(e^i)\cdot u^j_p\cdot v_r\}_{h:[m]\to \{\pm 1\},p\in [k], r\in [d]}$. We split the queries into $d$ slices. For each $r\in [d]$, the $r$-th slice consists of all queries $\{ h(e^i) \cdot u^j_p \cdot v_r \}_{h:[m]\to \{\pm 1\},p\in [k]}$. Similarly, for a mean-estimation algorithm $\calA:K^n\to \mathbb{R}^{mkd}$, we think of the output of $\calA$ as consisting of $d$ slices, each of dimension $mk$. Namely, the $r$-th slice is $\calA(x)_{*,*,r}$. 

We introduce a notion of partial score. For every $1\le r\le d$, define the $r$-partial score of a point $x\in K$ with respect to $q\in \mathbb{R}^{mkd}$ as
$$
\pscore^{(r)}(x; q) \coloneqq \left \langle x_{*,*,\le r} - \E_{x':type(x')=type(x)}[x'_{*,*,\le r}], q_{*,*,\le r}\right \rangle.
$$
That is, the $r$-partial score just sums up the contribution to the score from the first $r$ slices. Note that the partial score only depends on the first $r$ slices of the input and output. For this reason, given $x\in \mathbb{R}^{mk\cdot r}$ and $q\in \mathbb{R}^{mk \cdot r}$, we may slightly abuse notation by writing $\pscore^{(r)}(x;q) \coloneqq \pscore^{(r)}(\overline{x}, \overline{q})$ where $\overline{x}$ and $\overline{q}$ are arbitrary completion of $x$, $q$ into $mkd$-dimensional vectors.

\paragraph*{Fair algorithms.} We now define a class of query-releasing algorithms, that we call fair algorithms.

\begin{definition}\label{def:fair-algorithms}
    Let $\calA:K^n\to \mathbb{R}^{mkd}$ be a vector mean-estimation algorithm. Assume $\calA$ operates in $d$ stages. In the $r$-th stage, the algorithm outputs $\calA(x)_{*,*,r}$. Then, $\calA$ is informed of the $r$-th slice of $\theta$ (i.e.,~$\theta_{*,*,r}$). %The algorithm can use $\theta_{*,*,\le r}$ in the following stages.
    
    We say that $\calA$ is $\tau$-fair, if the following holds true for every stage $r\in [d]$.
    \begin{enumerate}
        \item Conditioning on $\calA(x)_{*,*,<r}$, the distribution of $\calA(x)_{*,*,r}$ only depends on the first $r$ slices of inputs and $\theta$. Moreover:
        \item For any input $x^j$, conditioning on $\calA(x)_{*,*,\le r} = q$ such that $\pscore^{(r)}(x^j;q) > \tau$, the distribution of $\calA(x)\mid (\calA(x)_{*,*,\le r}=q)$ is independent of the $(r+1)$-th to the $d$-th slice of $x^j$.
    \end{enumerate}
    A $d$-stage query-releasing algorithm $\calA:K^n\to [-1,1]^M$ is $\tau$-fair, if its induced vector mean estimation algorithm $\tilde{\calA}:K^n\to \mathbb{R}^{mkd}$ (via Section~\ref{sec:composition-surgery}) is $\tau$-fair.
\end{definition}

Let us digest Definition~\ref{def:fair-algorithms}. Item $1$ says that a fair algorithm, during any stage $r$, will not access the $(r+1)$-th to the $d$-th slice of its inputs. Item $2$ further asserts that the algorithm will not access $v_{>r}$ of a data point $x^j$ once its $r$-partial score has reached a certain threshold $\tau$. The name ``fair'' intuitively captures the fact that the algorithm uses its data in a somewhat fair manner. We observe an implication of Item 2: for the joint random variable $(x,\calA(x))$, conditioning on $x^j_{*,*,\le r}$ and $\calA(x)$ such that $\pscore^{(r)}(x^j,\calA(x))\ge \tau$, the slices $x^j_{*,*,>r}$ still have maximum uncertainty.

We also note that, since the algorithm is informed of $\theta_{*,*,\le r}$ after the conclusion of the $r$-th stage, it can compute $\pscore^{(r)}(x;\calA(x))$ by itself. Hence, in principle, we can convert any reasonable algorithm to a $\tau$-fair one. Here, having access to $\theta_{*,*,\le r}$ is crucial. Otherwise, the algorithm cannot compute $\pscore^{(r)}$ by itself, and it is not clear how to ask the algorithm to stop accessing an input $x^j$ at the ``right'' moment.

\paragraph*{Remark on the divergence-to-score lemma.} One may notice that the setup of Definition~\ref{def:fair-algorithms} is slightly different than the setup considered in Section~\ref{sec:composition}. Here, the algorithm publishes its output slice by slice. After the algorithm commits to its estimation for a slice, it is informed of the ``correct mean'' of that slice, a piece of information that may be utilized to answer future queries. Hence, it is natural to wonder whether the divergence-to-score connection (Proposition~\ref{prop:type-dependent-score-1}) still holds.

We now verify that Proposition~\ref{prop:type-dependent-score-1} holds in the new setup verbatim. Indeed, for every $i\in [mkd]$, consider an algorithm $\calB$ that has access to every $\theta_{j}$ other than $\theta_i$ and a bunch of input $(x^1,\dots, x^n)$. Suppose $\calB$ outputs a real value. We claim that
\begin{align}
\frac{\partial}{\partial \theta_i} \E_{x\sim D_{\theta}^n}[\calB(x, \theta_{-i})] = \sum_{j=1}^{n} \E_{x,\calB(x,\theta_{-i})}\left[ \calB(x, \theta_{-i}) \cdot \left( x^j_i - \E_{x'\sim D_\theta\mid type(x')=type(x^j)}[x'_i] \right) \right]. \label{equ:allow-algo-access-theta}
\end{align}
To prove the equation, we repeat the proof for Proposition~\ref{prop:type-dependent-score-1}, noting that allowing $\calB$ to access $\theta_{-i}$ has no effect on the derivation (since we do not differentiate $\theta_{-i}$, and the equation holds for every fixed $\theta_{-i}$ anyway).

Back to our context, let $\calA$ be the $d$-stage query-releasing algorithm. During the $r$-th stage, the algorithm can access its inputs and all $\theta_{*,*, < r}$. Since the algorithm cannot access the $r$-th slice of $\theta$, the $r$-th slice of the output is subject to a version of Equation~\eqref{equ:allow-algo-access-theta}. Adding up all slices concludes the proof of Proposition~\ref{prop:type-dependent-score-1} for our new setup.

Having confirmed Proposition~\ref{prop:type-dependent-score-1}, the score lower bound from Section~\ref{sec:composition-finish} applies to the algorithms considered in this section, so long as $\calA$ is accurate.

\subsection{Lower Bound Against Fair Algorithms}

In this section, we prove a lower bound against all fair algorithms.

\begin{theorem}\label{theo:lb-for-fair-algo}
    Let $n = \frac{c \sqrt{d} m}{\alpha^3 \log(1/\alpha)}$. Suppose $\calA:K^n\to [-1,1]^M$ is an algorithm for the workload matrix $A\in \{\pm 1\}^{M\times N}$. Assume $\calA$ is $\alpha$-accurate (w.r.t.~$\ell_{\infty}$-norm) and $\tau$-fair with $\tau \ge C\frac{\sqrt{d \log(1/\alpha)}}{m}$. Then, there is a $\theta\in V$ and an adversary $\calB$ which, upon seeing the output $\calA(x)$, crafts a query $q:K\to [-1,1]$ with the following on-average guarantee:
    $$
    \E_{x\sim D_\theta^n,\calA(x)}~\E_{q\gets \calB(\calA(x),\theta)} \left| \E_{j} [q(x^j)] - \E_{x\sim D_\theta}[q(x)] \right| \ge \Omega(\alpha).
    $$
\end{theorem}

\begin{proof}
    Choose $C>0$ to be a large constant. Let $\tilde{\calA}:K^n\to \mathbb{R}^{mkd}$ be the composition of $\calA$ with the post-processing described in Section~\ref{sec:composition-surgery}. From the argument of Section~\ref{sec:composition-finish} we see that, assuming $\calA$ (and hence $\tilde{\calA}$) is accurate, we obtain
    $$
    \E_{\theta\sim V} ~ \E_{x\sim D_\theta^n,\tilde{\calA}(x)} \left[ \sum_{j=1}^n \score(x^j;\tilde{\calA}(x)) \right] \ge \Omega(\frac{d}{\alpha^2}) .
    $$
    By averaging principle, we fix a $\theta$ for which the score lower bound holds true.

    For a realization of $(x,\tilde{\calA}(x))$, define a query $q:K\to [-1,1]$ by
    $$
    q(x) \coloneqq \max\left( -1, \min\left( 1, \score(x;\tilde{\calA}(x)) \cdot \frac{m}{2 C \sqrt{d \log(1/\alpha)}} \right)\right).
    $$

    \paragraph*{The average of $q$ on an independent point.} For an independent $x'\sim D_{\theta}$, we have that $|\score(x;\tilde{\calA}(x))| \le \frac{2 C\sqrt{d\log(1/\alpha)}}{m}$ with probability $1 - \alpha^{C^2}$. Hence, we have
    \begin{align}
    \left| \E_{x\sim D_\theta^n,\tilde{\calA}(x), x'\sim D_\theta}[q(x')] \right|
    &= \left| \E_{x,\tilde{\calA}(x),x'} \left[\frac{\score(x';\tilde{\calA}(x))\cdot m}{C \sqrt{d \log(1/\alpha)}} \right] \pm 2\cdot \Pr_{x'}\left[ |\score(x';\tilde{\calA}(x))| > \frac{C\sqrt{d\log(1/\alpha)}}{m} \right] \right| \notag \\
    &\le 2 \alpha^{C^2}. \label{equ:avg-q-independent-sample}
    \end{align}

    \paragraph*{The average of $q$ on in-sample points.} For each $j\in [n]$, let us examine the distribution of $\score(x^j;\tilde{\calA}(x))$. As before, we understand $x^j$ and $\tilde{\calA}(x)$ as consisting of $d$ \emph{slices}, each of dimension $m\times k$. Consider observing $x^j$ and $\tilde{\calA}(x)$ slice by slice. Before the partial score of $x^j$ exceeds $\tau$, we do not have good control over the growth of the score. However, once $\pscore^{(r)}(x^j,\tilde{\calA}(x))$ reaches $[\tau, \tau + \frac{1}{m})$\footnote{Since the score from each slice is bounded by $\frac{1}{m}$, there cannot be a sudden jump of partial score from one slice to the next.} for some $r\in [d]$, the future slices of $x^j$ and $\tilde{\calA}(x)$ become independent by the fair property of $\tilde{\calA}$. Write $x^j = (e^i\otimes u^j\otimes v)$. Conditioning on $e^i,u^j$ and $v_{\le r}$, we know that $v_{>r}$ is independent of $\tilde{\calA}(x)$ and each bit of $v_{>r}$ is independently sampled. As such, the contribution from the $(r+1)$-th to the $d$-th slice is $O(\frac{\sqrt{d-r}}{m})$-subgaussian. Overall, we see that $\score(x^j;\tilde{\calA}(x))$ is stochastically dominated by $\tau + O(\frac{\sqrt{d}}{m}) \cdot N(0, 1)$. Consequently, we have
    \begin{align}
    \Pr_{x,\tilde{\calA}(x)}\left[|\score(x^j;\tilde{\calA}(x))| \ge \tau + C\frac{\sqrt{d \log(1/\alpha)}}{m} \right] < \alpha^3 \notag
    \end{align}
    and
    $$
    \E_{x,\tilde{\calA}(x)}\left[|\score(x^j;\tilde{\calA}(x))| \middle| |\score(x^j;\tilde{\calA}(x))| \ge \tau + C\frac{\sqrt{d \log(1/\alpha)}}{m} \right] < O(C\cdot \frac{\sqrt{d\log(1/\alpha)}}{m}).
    $$
    Therefore, we obtain
    $$
    \begin{aligned}
    &~~~~ \E_{x,\tilde{\calA}(x)} \left[ \sum_{j=1}^{n} q(x^j) \cdot \frac{C\sqrt{d\log(1/\alpha)}}{m} \right] \\
    & \in \sum_{j=1}^{n} \E[\score(x^j;\tilde{\calA}(x))] \pm \Pr\left[|\score(x^j;\tilde{\calA}(x^j))| > \tau + \frac{C}{2}\frac{\sqrt{d\log(1/\alpha)}}{m} \right] \cdot O\left(\frac{C\sqrt{d\log(1/\alpha)}}{m}\right) \\
    & \in \sum_{j=1}^{n} \left( \E[\score(x^j;\tilde{\calA}(x))] \pm O\left(\frac{C\alpha^2\sqrt{d\log(1/\alpha)}}{m} \right) \right) \\
    &\ge \Omega\left(\frac{d}{\alpha^2}\right) - O\left(\frac{n\cdot C \alpha^3 \sqrt{d\log(1/\alpha)}}{m}\right) \\
    &\ge \Omega\left(\frac{d}{\alpha^2}\right) - O\left(d\right).
    \end{aligned}
    $$
    By scaling, this inequality is equivalently saying that
    \begin{align}
    \E_{x,\tilde{\calA}(x)} \E_{j\sim [n]} \left[ q(x^j) \right] \ge \Omega(\alpha). \label{equ:avg-q-in-sample}
    \end{align}
    Combining \eqref{equ:avg-q-in-sample} and \eqref{equ:avg-q-independent-sample} completes the proof.
\end{proof}

\subsection{Lower Bounds Against All Algorithms}

In this section, we lift Theorem~\ref{theo:lb-for-fair-algo} to prove a lower bound against all algorithms for answering adaptively generated linear queries. 

\subsubsection{Obfuscating Inputs}\label{sec:obfuscate-inputs}

Let $W = (\frac{\log(N)\log(M)}{\alpha^3})^{O(1)}$. We enlarge the universe $N$ by a factor of $W$. Now, we identify every point $x\in [N\times W]$ in the universe by a tuple $x = (w, e^i, u^{i'}, v)$. Here, $(e^i,u^{i'},v)$ has the same geometric interpretation as before (i.e., it is associated with a vector $e^i\otimes u^{i'}\otimes v\in \mathbb{R}^{mkd}$, and we newly introduce $w\in [W]$ as a ``name'' for $x$. 

For every $\theta\in \mathbb{R}^{mkd}$, we extend the definition to $D_{\theta}$ to the new setting: to sample $x = (w, e^i, u^{i'}, v)$, we sample the tuple $(e^i,u^{i'},v)$ as before, and sample a name $w\in [W]$ uniformly at random. The idea is, with our setting of $W$, we ensure that for an i.i.d.~data set $x\sim D_{\theta}^{n}$, with probability $1-\frac{1}{W^{\Omega(1)}}$, every data point $x^j$ has a distinct name.

Now, consider a set of random mappings of the form:
$$
\Pi_r: [W]\times [m]\times [k]\times \left(\{\pm 1\}^{0}\cup \{\pm 1\}^{1} \cup \dots \cup \{\pm 1\}^{r-1}\right) \to \{\pm 1\}.
$$
Given a data point $x = (w, e^i,u^{i'}, v)$, let its $\Pi$-\emph{obfuscation} be a point $\Pi(x) = (w, e^i, u^{i'}, v')$, defined as follows. As the notation suggests, the name and type of $\Pi(x)$ are the same as $x$. However, for every $r\in [d]$, we define $v'_r = \Pi_r(w, i,i', v_{<r}) \oplus v_r$ where $\oplus$ denotes the XOR of two bits. Since $\Pi$ only changes the $v$-part of $x$, we write $v(\Pi, x)$ to denote the ``$v$''-part of $\Pi(x)$. Note that the obfuscation is invertible. Write $\Pi^{-1}$ to be its inverse. Namely, $\Pi^{-1}$ is such that $\Pi^{-1}(\Pi(x)) = x$ for every $x = (w, e^i,u^{i'}, v)$. Finally, for any distribution $D$ over $[N\times W]$, let $\Pi(D)$ be the distribution of $\Pi(x)$ where $x$ is drawn from $D$.
% Now, consider a random mapping of the form:
% $$
% \Pi: [W]\times [m]\times [k]\times \left(\{\pm 1\}^{0}\cup \{\pm 1\}^{1} \cup \dots \cup \{\pm 1\}^{r-1}\right) \to \{\pm 1\}.
% $$
% Given a data point $x = (w, e^i,u^{i'}, v)$, let its $\Pi$-\emph{obfuscation} be a point $\Pi(x) = (w, e^i, u^{i'}, v')$, defined as follows. As the notation suggests, the name and type of $\Pi(x)$ are the same as $x$. However, for every $r\in [d]$, we define $v'_r = \Pi(w, i,i', v_{<r}) \oplus v_r$ where $\oplus$ denotes the XOR of two bits. Since $\Pi$ only changes the $v$-part $x$, we write $v(\Pi, x)$ to denote the ``$v$''-part of $\Pi(x)$. Note that the obfuscation is invertible. Write $\Pi^{-1}$ to be its inverse. Namely, $\Pi^{-1}$ is such that $\Pi^{-1}(\Pi(x)) = x$ for every $x = (w, e^i,u^{i'}, v)$. Finally, for any distribution $D$ over $[N\times W]$, let $\Pi(D)$ be the distribution of $\Pi(x)$ where $x$ is drawn from $D$.

% Suppose $(x^1,\dots, x^n)$ is a data set of $n$ points with distinct names. Consider its $\Pi$-obfuscation $(\Pi(x^1),\dots, \Pi(x^n))$ where $\Pi$ is randomly chosen. We see that while $\Pi$ preserves the type of each point, it completely hides the ``$v$''-part of the data points. Moreover, assuming $(x^1,\dots, x^n)$ have distinct names, the obfuscation of different points is independent.

\subsubsection{Making an ADA Algorithm Fair}\label{sec:algorithm-made-fair}

We describe a reduction to make an ADA algorithm ``fair''. The idea is to compare the following two experiments: 
\begin{itemize}
    \item For an unknown $\theta$, draw $(x^1,\dots, x^n)\sim D_{\theta}^n$ and perform ADA with respect to $D_{\theta}$.
    \item For an unknown $\theta$ and \emph{unknown} $\Pi$, draw $(\Pi(x^1),\dots, \Pi(x^n))\sim \Pi(D_{\theta})^n$ and perform ADA with respect to $\Pi(D_{\theta})$.
\end{itemize}
If $\Pi$ is known, the two tasks are equivalent as $\Pi$ is a bijection. However, if $\Pi$ is unknown to the algorithm, from $(\Pi(x^1),\dots, \Pi(x^n))$, the algorithm only learns the type of each data point. The proof strategy is to design an interaction between the algorithm and the adversary, through which the algorithm gradually learns each slice of its inputs. By designing the adversary properly, we can ensure that the algorithm behaves fairly.

Let $\calA$ be the ADA algorithm. We design the adversary below. 
\begin{itemize}
\item At the start, some $\theta$ is chosen, and a random $\Pi$ is generated. A data set $(x^1,\dots, x^n)\sim D_\theta^n$ is drawn. $\calA$ receives $\Pi(x^1),\dots, \Pi(x^n)$. Note that this is equivalent to drawing $\Pi(x^1),\dots, \Pi(x^n)\sim \Pi(D_\theta)^n$ in the first place.

\item The adversary interacts with $\calA$ for $d$ stages. In the $r$-th stage, consider the set of queries
$$
H^{r} = \{ h(e^i)\otimes u^j_p \otimes v_r \}_{h:[m]\to \{\pm 1\}, p\in [k]}
$$
We understand $H^{r}$ as an ensemble of statistical queries w.r.t.~$D_\theta$. We assume $H^r$ is publicly known to both $\calA$ and the adversary. The equivalent of $H^r$ with respect to $\Pi(D_{\theta})$ is given by
$$
H^{\Pi, r} = \left\{ g(w,e^i,u^{j},v) \coloneqq h(e^i)\cdot u^j_p \cdot v(\Pi^{-1}, (w,e^i,u^j,v))_{r}\right\}_{{h:[m]\to \{\pm 1\}, p\in [k]}}.
$$
Namely, given a query $g\in H^r$, there is a corresponding query $g'\in H^{\Pi, r}$ defined by $g'(x) = g(\Pi^{-1}(x))$.

For a moment, let us consider letting the adversary send the queries $H^{\Pi, r}$ to $\calA$ in the $r$-th stage. By comparing $H^{r}$ with $H^{\Pi, r}$, $\calA$ can learn $\Pi(w,i,j,v_{<r})$ for every $w,i,j,v_{<r}$. From this, it learns the $r$-th slice of $x^1,\dots, x^n$. However, it is still clueless about the $(r+1)$-th to $d$-th slices of the inputs.

\item In each stage, the adversary receives responses from $\calA$. These responses are answers to $H^{\Pi, r}$ with respect to $\Pi(D_\theta)$. They can be directly translated into answers to $H^r$ with respect to $D_{\theta}$. By further running the post-processing from Section~\ref{sec:composition-surgery}, the adversary obtains a vector $q_{*,*,r}\in \mathbb{R}^{mk}$ that approximates $\E[D_\theta]_{*,*,r}$. Concatenating all $q_{*,*,r'}$ for $r'\le r$ would allow for tracking $\pscore^{(r)}(x;q)$ for every $x\in [N\times W]$. 

Now we come to a crucial part: At the start of each stage $r$, for each $x\in [N\times W]$, if $\pscore^{(r-1)}(x;q) > \tau$, we mark $x$ as \emph{compromised}. As soon as a point $x$ is compromised, the adversary will \emph{not} reveal the correct evaluation of $g(\Pi(x))$ in the future. Namely, for each $g\in H^{\Pi,r}$, define $g'$ as $g'(\Pi(x)) = g(\Pi(x))$ if $x$ not compromised, and $g'(\Pi(x)) = 1$ otherwise. Let $\tilde{H}^{\Pi,r}$ be the collection of all the $g'$s. Instead of giving $\calA$ the query set $H^{\Pi,r}$, the adversary actually sends the set $\tilde{H}^{\Pi,r}$.

As we will argue in a moment, queries from $\tilde{H}^{\Pi,r}$ have distributional means similar to their counterparts in $H^{\Pi,r}$. Hence, although the adversary sends queries from $\tilde{H}^{\Pi,r}$ and receives responses about them, it can post-process these responses just as if they were queries about $H^{\Pi,r}$.
\end{itemize}

\subsubsection{Analysis}\label{sec:fair-reduction-analysis}

The last section presented our adversary design. We now discuss its correctness. The interaction between $\calA$ and the adversary as a whole constitutes a query-releasing procedure $\mathcal{P}$ with respect to $D_{\theta}$. The procedure $\mathcal{P}$ takes $(x^1,\dots, x^n)$ as input and outputs $q\in \mathbb{R}^{mkd}$. We now argue that $\mathcal{P}$ is both accurate and fair.

\paragraph*{Fairness.} We start by analyzing fairness. Item $1$ of Definition~\ref{def:fair-algorithms} is satisfied, as the adversary reveals $\Pi$ ``bit-by-bit'' by our design. Regarding Item $2$, note that if a point $x = (w,e^i,u^j,v_{\le r-1}\circ v_{\ge r})$ is compromised in the $(r-1)$-th round, so is every point of the form $x'=(w,e^i,u^j,v_{\le r-1}\circ v'_{\ge r})$. Take one such $x'$, and we compare the behavior of the algorithm on, e.g., $(\Pi(x),\Pi(x^2),\dots, \Pi(x^n))$ versus $(\Pi(x'), \Pi(x^2),\dots, \Pi(x^n))$. Over a random $\Pi$, we know $\Pi(x)$ and $\Pi(x')$ are identically distributed, even after conditioning on the first $(r-1)$ slices of $\Pi$. Should $x$ be not compromised, the algorithm could have distinguished between $\Pi(x)$ and $\Pi(x')$ by examining $g(\Pi(x))$ and $g(\Pi(x'))$ for some query $g$ from $\tilde{H}^{\Pi,\ge r}$. However, since $x$ has been compromised, we have $g(\Pi(x))=g(\Pi(x')) = 1$ for every $g\in \tilde{H}^{\Pi,\ge r}$. Therefore, from the $r$-th stage onward, the algorithm cannot distinguish between the input being $x$ or $x'$. That is, the algorithm behaves the same on every $x'$ of the form $(w,e^i,u^j,v_{\le r-1}\circ v'_{\ge r})$, meaning that the $r$-th to $d$-th slice of the input is independent of the algorithm's output.

If the names (the ``$w$'' part) of $x^1,\dots, x^n$ are distinct (which happens with probability $1-\frac{1}{W^{\Omega(1)}}$), their random masks do not interfere with each other, and we can apply the argument for each input separately. This shows the procedure is $\tau$-fair with probability $1-\frac{1}{n^{\Omega(1)}}$.

\paragraph*{Accuracy.} We now prove the procedure is accurate with respect to $D_{\theta}$. The promise of $\calA$ says that its responses are $\alpha$-accurate for the query family $\tilde{H}^{\Pi,*}$. If we can argue that these responses are $(1+o(1))\alpha$-accurate to the closely related query family $H^{\Pi,*}$, then the adversary can just directly translate these responses into accurate responses for $H^{*}$ (with respect to $D_{\theta}$). Then, the accuracy of $q=\mathcal{P}(x^1,\dots, x^n)$ follows from accurate answers to $H^*$ together with the post-processing from Section~\ref{sec:composition-surgery}.

Indeed, every $g\in H^{\Pi,r}$ is associated with a query $g'\in \tilde{H}^{\Pi,r}$. By definition, we have
$$
\left| \E_{x\sim D_{\theta}}[g(x)] - \E_{x\sim D_{\theta}}[g'(x)] \right| \le \Pr_{x}[x \text{ is compromised before stage $r$}].
$$
Therefore, it suffices to prove that the fraction of compromised points is bounded by $\alpha$ across all stages. At a stage $r\in [d]$, conditioning on $q_{*,*,<r}$, we have
$$
\Pr_{x}[x \text{ is compromised before stage $r$}] = \Pr_{x}\left[\exists r'<r, \pscore^{(r')}(x,q) \le \tau \right].
$$
Take $\{ \pscore^{r'}(x,q) \}_{r' < r}$ as a random walk with $(r-1)\le d$ steps. The movement of each step has zero mean, and the length of each step is bounded by $\frac{1}{m}$. Since $\tau = \frac{C\sqrt{d\log(1/\alpha)}}{m}$, with probability $1-\alpha^{\Omega(C^2)}$ over $x$, the random walk stays below the threshold $\tau$ across the $r-1$ steps. This means the fraction of compromised inputs is always below $\alpha^{\Omega(C)}$ across the $d$ stages, regardless of the partial output $q_{*,*,\le r}$. This completes the accuracy analysis.

\subsubsection{Proof of Theorem~\ref{theo:intro-ada-lb}}\label{sec:ada-lb-finish}

We are now ready to prove one of the main results of the paper, Theorem~\ref{theo:intro-ada-lb}. We formulate a formal version of Theorem~\ref{theo:intro-ada-lb} below.

\begin{theorem}\label{theo:ada-lb-formal}
Let $N,M\ge 1$ and $\alpha\in (0, 1/10)$ be such that $N\ll 2^M$, $M\ll 2^N$ and $\frac{1}{\alpha}\ll \min(N,M)$. Let $\calA$ be an algorithm for answering statistical linear queries over the domain $[N]$. Suppose $\calA$ operates on at most $o(\frac{\sqrt{\log(N)}\log(M)}{\alpha^3 \log(1/\alpha)})$ samples. Then, there is a $O(\log N)$-round adaptive attack against $\calA$, which sends at most $m$ queries to $\calA$ and makes $\calA$ fail to be either sample- or distribution-accurate within error $\alpha$ on at least one query. The attack succeeds with probability $\Omega(\alpha)$.
\end{theorem}

Let us quickly review the relevant parameter settings. We have $d=\Theta(\log N), m = \Theta(\log M)$, $k = \Theta(\frac{1}{\alpha^2})$ such that $2^d\cdot mk = N$ and $2^m\cdot kd = M$. In Section~\ref{sec:obfuscate-inputs}, we also set up a new parameter $W=(\frac{\log(N)\log(M)}{\alpha^3})^{O(1)}$ and use it to design an ``obfuscating scheme''. To keep notational consistency with previous sections, here we choose to work with a slightly large domain $[N\times W]$ in proving Theorem~\ref{theo:ada-lb-formal}, noting that this does not change our conclusion as $\log(N) = \Theta(\log(NW))$. Having verified the parameter consistency, the upcoming proof will use the constructions and designs from prior sections (such as the set $K$, the mapping $\Pi$, the matrix $A\in \{\pm 1\}^{M\times N}$, the distribution $D_{\theta}$, etc.) without further notice.

\begin{proof}
Suppose $\mathcal{A}$ is an ADA algorithm working over the universe $[N\times W]$, operating on $n\le \frac{c\sqrt{\log(N)}\log(M)}{\alpha^3}$ points. Consider an adversary $\calB$ interacting with $\mathcal{A}$ as in Section~\ref{sec:algorithm-made-fair}. Denote by $\mathcal{P}$ the whole procedure. $\mathcal{P}$ takes as input $x^1,\dots, x^n\sim D_\theta^n$ and the description of $\Pi$. It simulates the interaction between $\mathcal{A}$ and $\mathcal{B}$ and outputs an estimate of $\E[D_\theta]$, denoted by $q\in \mathbb{R}^{mkd}$. We have argued that $\mathcal{P}$ is fair w.r.t.~$(x^1,\dots, x^n)$ with probability $1-\frac{1}{W^\Omega(1)}$.

Depending on how often the output of $\calA$ is accurate, we consider two cases.

\paragraph*{Case 1.} Let $c > 0$ be sufficiently small. If, with probability $c\cdot \alpha$, $\calA$ fails to return $\alpha$-accurate answer (w.r.t.~$\Pi(D_\theta)$) to at least one query, the adversary $\mathcal{B}$ is the desired attack to $\calA$.

\paragraph*{Case 2.} Now we assume that with probability $1-c\alpha$, all outputs of $\mathcal{A}$ are $\alpha$-accurate w.r.t.~$\Pi(D_\theta)$. We argue that the output of $\mathcal{P}$ is close to $\E[D_\theta]$ on average. First, from Section~\ref{sec:fair-reduction-analysis}, we see that $\alpha$-accurate responses from $\calA$ induce an output $q$ of $\mathcal{P}$ such that $\| q - \E[D_{\theta}] \|\le O(\alpha kd)$. Second, regardless of $\mathcal{A}$ being accurate or not, the output $q\gets \mathcal{P}(x)$ always satisfies that $\| q\|_1 \le O(\sqrt{k} d)$. To see this, for every $i\in [m]$ and $r\in [d]$, there is $\lambda\in [-1,1]^{k}$ such that
$$
\| q_{i,*,r} \|_1 = \frac{1}{m} \left\| \frac{1}{k} \sum_{j=1}^{k} \lambda_j u^j \right\|_1 \le \frac{\sqrt{k}}{m} \left\| \frac{1}{k} \sum_{j=1}^{k} \lambda_j u^j \right\|_2 \le \frac{\sqrt{k}}{m}
$$
Summing over all $i,r\in [d\times m]$ verifies the claim. With this in mind, we calculate
$$
\left\| \E[\mathcal{P}(x)] - \mu_\theta \right\|_1 \le O(\alpha kd) + \Pr[\mathcal{A} \text{ not accurate}] \cdot O(\sqrt{kd}) \le O(\alpha k d).
$$
Hence, $\mathcal{P}$ is both accurate and natural, placing itself under the regime of Theorem~\ref{theo:lb-for-fair-algo}. By Theorem~\ref{theo:lb-for-fair-algo}, we can fix a $\theta$ and run the following attack. For every output $q$ of $\mathcal{P}$, we can define a query $\phi$ accordingly, such that
\[
\E_{x\sim D_\theta^n} \E_{q\sim \mathcal{P}(x), \phi} \left| \E_j[\phi(x^j)] - \E_{D_\theta}[\phi(x)] \right| \ge \Omega(\alpha).
\]
Since the deviation of the sample-mean from the distributional mean is always bounded by $1$, by the reverse Markov's inequality, with probability $\Omega(\alpha)$, it holds that: 
$$
\left| \E_j[\phi(x^j)] - \E_{D_\theta}[\phi(x)] \right| \ge \Omega(\alpha).
$$
For any $\phi$ for which the above holds, define $\phi'(x) \coloneqq \phi(\Pi^{-1}(x))$ accordingly. Then, it follows that
$$
\left| \E_j[\phi'(\Pi(x^j))] - \E_{x\sim \Pi(D_\theta)}[\phi'(x)] \right| \ge \Omega(\alpha).
$$
We see that $\phi'$ witnesses a mismatch between the data $(\Pi(x^1),\dots, \Pi(x^n))$ and the distribution $\Pi(D_\theta)$, rendering itself a desired attack query against $\calA$.

\end{proof}

\subsection{On Removing the Sample-Accurate Assumption}

In this section, we present several lower bounds against algorithms that are \emph{only} required to be accurate with respect to the distribution. The price we pay, however, is that the bound usually becomes smaller by a factor of $\frac{1}{\alpha}$.

\paragraph*{Lower bounds in the many-query regime.} We record relevant parameters here for quick reference: $N$ is the size of the universe, $M$ the number of queries, and $\alpha$ the desired accuracy. We have set $d = \Theta(\log(N)), m = \Theta(\log(M)), k = \frac{1}{\alpha^2}$ so that $2^m\cdot kd = M$ and $2^{d}\cdot km = N$. We have constructed a set of vectors $K=\{ (e^i\otimes u^j\otimes v) \}$ of size $|K| = mk2^d = N$, and an associated query matrix $A\in \{\pm 1\}^{M\times N}$.

As before, we start by analyzing the class of fair algorithms.

\begin{proposition}\label{prop:lb-for-fair-algo-distribution}
    Let $n = \frac{c \sqrt{d} m}{\alpha^2\log(1/\alpha)}$ and $\tau = \frac{C\sqrt{d\log(1/\alpha)}}{m}$. There is no $\tau$-fair algorithm $\calA:K^n\to [-1,1]^M$ that can answer all queries for the workload matrix $A\in \{\pm 1\}^{M\times N}$ to within $\ell_\infty$ generalization error $\alpha$ with probability $1-o(1)$.
    % is an algorithm for the workload matrix $A\in \{\pm 1\}^{M\times N}$. Assume $\calA$ is $\alpha$-accurate (w.r.t.~$\ell_{\infty}$-norm) and $\tau$-fair with $\tau \ge C\frac{\sqrt{d \log(1/\alpha)}}{m}$. Then, there is a $\theta\in V$ and an adversary $\calB$ which, upon seeing the output $\calA(x)$, crafts a query $q:K\to [-1,1]$ with the following on-average guarantee:
    % $$
    % \E_{x\sim D_\theta^n,\calA(x)}~\E_{q\gets \calB(\calA(x),\theta)} \left| \E_{j} [q(x^j)] - \E_{x\sim D_\theta}[q(x)] \right| \ge \Omega(\alpha).
    % $$
\end{proposition}

\begin{proof}
    Suppose for contradiction that such an $\calA$ exists. Let $\tilde{\calA}$ be the composition of $\calA$ with the post-processing of Section~\ref{sec:composition-surgery}. Recall our parameter setting that $k = \frac{c}{\alpha^2}$ for some small $c$. Then, for any $\theta$ and its induced distribution $D_\theta$, we have
    $$
    \begin{aligned}
    &~~~~ \| \E_{x\sim D_\theta^n,\tilde{\calA}(x)}[ \tilde{\calA}(x)] - \E[D_\theta] \|_1 
    &\le O(\alpha kd) + \Pr[\mathcal{A} \text{ not accurate}] \cdot O(\sqrt{k} d) \\
    &\le O(\alpha kd).
    \end{aligned}
    $$
    In particular, this means $\tilde{\calA}$ is subject to the score lower bound from Section~\ref{sec:composition-finish}. Namely,
    $$
    \E_{\theta} \E_{x,\tilde{\calA}(x)} \left[  \sum_{j=1}^n \score(x^j;\tilde{\calA}(x^j)) \right] \ge \Omega(\frac{d}{\alpha^2}).
    $$
    On the other hand, assuming $\mathcal{A}$ is $\tau$-fair, we have
    $$
    \E_{\theta}\E_{x,\tilde{\calA}(x)}[\score(x^j;\tilde{\calA}(x)] \le \tau + O(\frac{C\sqrt{d \log(1/\alpha)}}{m}) \le O(\frac{\sqrt{d \log(1/\alpha)}}{m}).
    $$
    If $n < \frac{c\sqrt{d} m}{\alpha^2 \log(1/\alpha)}$, the score lower bound and upper bound becomes contradictory. This completes the proof.
\end{proof}

Section~\ref{sec:algorithm-made-fair} has described a reduction, roughly saying the following. Let $\mathcal{A}$ be an  ADA algorithm working over the domain $[N\times W]$. One can combine $\calA$ with the reduction procedure and obtain a $\tau$-fair query releasing algorithm, which works over a smaller universe $[N]$, enjoys the same sample complexity, and incurs a slightly higher error (higher by a $(1+o(1))$ factor). Furthermore, the reduction interacts with $\calA$ for at most $d$ adaptive rounds. Since all fair algorithms have been ruled out by Proposition~\ref{prop:lb-for-fair-algo-distribution}. We can thus combine the reduction with Proposition~\ref{prop:lb-for-fair-algo-distribution} and prove the following theorem.
\begin{theorem}\label{theo:lb-for-all-algo-distribution}
    Let $N,M\ge 1$ and $\alpha \in (0,1/10)$ be such that $2^N\gg M$, $2^M\gg N$ and $\frac{1}{\alpha}\ll \min(N,M)$. For some $n=\Theta(\frac{\sqrt{\log(N)} \log(M)}{\alpha^2 \log(1/\alpha)})$, the following is true: any ADA algorithm, over the universe $[N]$ and operating on $n$ samples, cannot answer $M$ adaptively generated statistical queries to within generalization error $\alpha$.

    Furthermore, there is an attack against $\mathcal{A}$, which breaks its accuracy within $O(\log N)$ rounds of adaptivity with constant probability. The attack sends at most $M$ queries.
\end{theorem}

% We recall our parameter setting for verifying Theorem~\ref{theo:lb-for-all-algo-distribution} quantitatively: Throughout the discussion, we set $d=\Theta(\log N), m = \Theta(\log(M)), k = \Theta(\frac{1}{\alpha^2})$ and $W=n^{O(1)} = \polylog(N,M)$. Hence, the universe size $N\times W$ is bounded by $\poly(N)$. Writing our lower bound as a function of $N, M, \alpha$ only, we obtain Theorem~\ref{theo:lb-for-all-algo-distribution}.

\begin{remark} 
If we ``unpack'' the proof of Theorem~\ref{theo:lb-for-all-algo-distribution}, we will find out that the attack strategy is similar to prior works \cite{HardtU14,SteinkeU15}. Namely, we assign a ``score'' to every point in the universe. We deem any point with a large score likely in the data set. As such, we remove their contribution to future queries. At a certain moment, almost all data points of the algorithm are exposed. Consequently, the algorithm fails to evaluate new queries accurately. However, due to the limitation of traditional fingerprinting lemmas (which required a Boolean hypercube structure), prior works cannot push the technique to the many-query small-universe regime. In contrast, our geometric fingerprinting technique allows for working with a richer class of geometric structures, enabling us to make further progress in understanding the limitations of efficient adaptive data analysis.
\end{remark}

\paragraph*{Lower bounds in the few-query regime.} So far, we are interested in the regime where the number of queries is large (compared with log-universe-size). For a smaller number of queries, we establish the following theorem.

\begin{theorem}\label{theo:lb-for-fewer-queries}
Let $n = \frac{c\sqrt{m}}{\alpha \log(1/\alpha)}$ and $N = 2^{O(\alpha^2 m)}$. Let $\calA$ be an ADA algorithm over the universe $[N]$ where $N=2^{O(\alpha^2 m)}$. Suppose $\calA$ receives only $n$ samples. Then, it cannot answer $m$ adaptively generated statistical queries within generalization error $\alpha$.

Furthermore, there is an attack against $\mathcal{A}$ that breaks its accuracy within $O(\alpha^2 m)$ rounds of adaptivity with constant probability. The attack sends at most $m$ queries.
\end{theorem}

Theorem~\ref{theo:lb-for-fewer-queries} recovers the main information-theoretic lower bound of \cite{SteinkeU15}. Our theorem is slightly stronger in the sense that our attack breaks the algorithm within $\alpha^2 k$ rounds of adaptivity, where in each round, the attacker sends $\frac{1}{\alpha^2}$ queries in a batch. Prior constructions require full adaptivity, and they send queries one by one.

\begin{proof}[Proof sketch.]
    We lift the query-releasing lower bound of Theorem~\ref{theo:two-way-marginal-lb} (whose proof appeared in Section~\ref{sec:composition-two-way}) to an ADA lower bound.

    In particular, recall that $k = \frac{c}{\alpha^2}$. For the ensemble of vectors
    $$
    K = \{u^j\otimes v: v\in \{\pm 1\}^{d} \} \subseteq \mathbb{R}^{kd}
    $$
    and its associated query matrix $A\in \{\pm 1\}^{kd\times (k2^d)}$, we have shown that any $\alpha$-accurate (in $\ell_\infty$ metric w.r.t.~$D_{\theta}$) algorithm $\mathcal{A}$ satisfies that 
    $$
    \E_{\theta} \E_{x,\calA(x)} \left[ \sum_{j=1}^{n} \score(x^j;\calA(x)) \right] \ge \Omega( \frac{d}{\alpha^2}).
    $$
    If we further assume that $\calA$ is $\tau=C\sqrt{d\log(1/\alpha)}$-fair w.r.t.~the slicing of $\mathbb{R}^{kd}$ into $d$ slices, then we have
    $$
    \E_{\theta}\E_{x,\calA(x)} \left[\score(x^j;\calA(x)) \right] \le \tau + O(\sqrt{d\log(1/\alpha)}) \le O(\sqrt{d\log(1/\alpha)}).
    $$
    The score upper and lower bounds give a lower bound of $n\ge \Omega\left(\frac{\sqrt{d}}{\alpha^2 \log(1/\alpha)}\right)$. In terms of the number of queries $m=kd$, this is an $\Omega(\frac{\sqrt{m}}{\alpha \log(1/\alpha)})$ lower bound. The universe size is $k2^d = 2^{O(d)} = 2^{O(\alpha^2 m)}$.

    We have shown a lower bound against fair algorithms. We can use a version of the reduction from Section~\ref{sec:algorithm-made-fair} to obtain a lower bound against all algorithms. The reduction enlarges the universe size by a factor of $W = n^{O(1)}$, which is negligible for a typical parameter regime where $2^{\alpha^2 m}\gg n$.
\end{proof}

% Let $n = \frac{c\sqrt{d}m}{\alpha^2\log(1/\alha)}$. If there were an ADA algorithm over the universe $[N\times W]$ that can answer $M = 2^{\Theta(m)}$ queries, then from the reduction we obtain a natural algorithm.

\section{Lower bounds on Random Query Releasing via Fingerprinting}

In this section, we prove the following lower bound for a set of random queries.

\begin{theorem}\label{theo:random-query-lb}
Let $N,d\ge 0$ be two integers such that $N\gg d$ and $\log(N) \le o(d)$. Consider a random matrix $A\in\{\pm 1\}^{d\times N}$ where each entry of $A$ is independently set to $\pm 1$ with equal probability.

Then, with probability $1-o(1)$ over $A$, the following is true for all $\alpha < \frac{c \sqrt{\log N}}{\sqrt{d}}$: for every $\eps \in (0, 1)$ and $\delta < \frac{\alpha}{d}$, any $(\eps,\delta)$-DP algorithm $\calA$ for query releasing with workload matrix $A$ needs at least $\Omega\left( \frac{\sqrt{d}}{\eps \alpha} \right)$ samples to achieve a mean-squared error of $\alpha^2 d$.
\end{theorem}

We will prove Theorem~\ref{theo:random-query-lb} via our geometric fingerprinting framework. Here, both the upper and lower bounds on the ``score'' are not straightforward, and we will use randomness of $A$ in an essential way to establish both bounds.

We also highlight that Theorem~\ref{theo:random-query-lb} is the first application of our framework on a set $K$ that does not appear to contain a large hypercube. In contrast, the lower bounds presented in Sections~\ref{sec:composition} and \ref{sec:lb-ada} all depend heavily on an embedded hypercube structure.

\subsection{Basic Facts on Rademacher Sums}

We need tight control on the tail bounds of Rademacher sums. To begin with, the following tail upper bound is a direct consequence of the Hoeffding inequality.

\begin{lemma}\label{lemma:rad-tail-ub}
    Let $a\in \mathbb{R}^d$. Then, for every $t\ge 1$, it holds that
    $$
    \Pr_{x\sim \{\pm 1\}^{d}}[\langle a, x\rangle \ge t\|a\|_2] \le \exp(-\frac{t^2}{2}).
    $$
\end{lemma}

We also need the following tail \emph{lower} bound.

\begin{lemma}\label{lemma:rad-tail-lb}
    Let $a\in \mathbb{R}^d$ be such that there are at least $d/2$ coordinates $i\in [d]$ with $|a_i| \ge \frac{1}{5 \sqrt{d}} \|a\|_2$. Then, there is an absolute $c > 0$ such that for every $t\in [0, c\sqrt{d})$, it holds that
    $$
    \Pr_{x\sim \{\pm 1\}^{d}}[\langle a, x\rangle \ge t\|a\|_2] \ge \exp(-c t^2).
    $$
\end{lemma}

We include a proof of Lemma~\ref{lemma:rad-tail-lb} below. To start, for any vector $a\in \mathbb{R}^{d}$ and $t>0$, define
$$
K_{1,2}(a,t) = \inf \{ \| a' \|_1 + t \|a''\|_2 : a'+a'' = a\}.
$$
For small $t$, the quantity $K_{1,2}$ behaves like $t \|a\|_2$ while for large $t$ (in particular for $t > \sqrt{d}$), the quantity converges to $\|a\|_1$. The following lemma instantiates this intuition for all ``good'' vectors, which are of interest to us.

\begin{lemma}\label{lemma:decompose-good-a}
    There is a small constant $c > 0$ for which the following is true. Suppose $a\in \mathbb{R}^d$ is such that there are at least $d/2$ coordinates $i\in [d]$ with $|a_i| \ge \frac{1}{5 \sqrt{d}} \|a\|_2$. Then, for $t < c\sqrt{d}$, it holds that
    \[
    t\|a\|_2 \ge K_{1,2}(a, t) \ge c\cdot t\cdot \|a\|_2.
    \]
\end{lemma}

\begin{proof}
    The upper bound on $K_{1,2}$ is obvious. We establish the lower bound here. Suppose $K_{1,2}(a,t) = \|a'\|_1 + t \|a''\|_2$. Among all the coordinates $i\in[d]$ such that $|a_i|\ge \frac{1}{5\sqrt{d}}\|a\|_2$, if there are half of them such that $a'_i > |a_i|/2$, then it is clear that $\|a'\|_1 \ge \frac{1}{10\sqrt{d}} \cdot \frac{d}{4}\cdot \|a\|_2 > c t \|a\|_2$ (recall that $t < c\sqrt{d}$).

    Otherwise, there are at least half of $i$'s with $|a''_i|\ge |a_i|/2$. Then, we have
    $$
    \|a''\|_2^2 \ge \frac{d}{4} \cdot \left( \frac{1}{10\sqrt{d}} \right)^2 \cdot \|a\|_2^2 > c^2 \|a\|_2^2.
    $$
    Consequently, we have $t \|a''\|_2 \ge c t \|a\|_2$ as desired.
\end{proof}

We also need the following anti-concentration inequality from \cite{dist-of-rademacher}.

\begin{lemma}[\cite{dist-of-rademacher}]\label{lemma:rademacher-with-K12}
There is a constant $c > 0$ such that for every $a\in \mathbb{R}^d$, we have
$$
\Pr_{x\sim \{\pm 1\}^d}[\langle x, a\rangle \ge c^{-1} K_{1,2}(a,t) ] \ge c^{-1} e^{-ct^2}.
$$
\end{lemma}

Now, it is clear that Lemma~\ref{lemma:rad-tail-lb} is a direct consequence of Lemma~\ref{lemma:rademacher-with-K12} combined with Lemma~\ref{lemma:decompose-good-a} (The constant ``$c$'' appearing in the three lemma statements are not necessarily the same).

\subsection{Lower Bounding the Score}

\paragraph*{Expanding vector family.} We state the following ``expanding'' property on a collection of vectors.
\begin{definition}\label{def:expanding-vectors}
    Let $d\in \mathbb{N}$. A collection of $d$-dimensional vectors $S = \{s^1,\dots, s^N\}$ is $(r,\eta)$-expanding, if the following is true:
    with probability $1-\frac{1}{d^{5}}$ over $\theta\sim \partial B_2(0,r)$, the exponential tilt $D_{\theta}$ of $S$ satisfies that
    $$
    \E_{v\sim D_{\theta}(S)}[\langle v, \theta\rangle ] \ge \eta.
    $$
    A matrix $A\in \{\pm 1\}^{d\times N}$ is $(r,\eta)$-\emph{expanding}, if the ensemble of column vectors of $A$ is $(r,\eta)$-expanding.
\end{definition}

Looking ahead, an expanding family of vectors gives a query matrix $A$ for which we can prove the score lower bound (c.f.~Section~\ref{sec:fingerprinting-overview}). The following proposition states that a random query matrix is expanding with desired parameters with high probability.

\begin{proposition}\label{prop:random-is-expanding}
     For all sufficiently small $c > 0$, the following is true. Suppose $N\gg d$ and $\log(N) \le o(d)$. Let $A\sim \{\pm 1\}^{d\times N}$ be chosen uniformly at random. With probability $1-o(1)$, the matrix $A$ is $(c\sqrt{\log N}, \Omega_c(\log(N)))$-expanding.
\end{proposition}

\begin{proof}
Let $r = c\sqrt{\log N}$. We claim that
\begin{align}
\Pr_{\substack{\theta\sim \partial B_2(0,r) \\ A\sim \{\pm 1\}^{d\times N}}} \left[ \E_{v\sim D_{\theta}(A)}[\langle v, \theta\rangle ] \ge \eta \right] \le \frac{1}{d^{10}}. \label{equ:random-theta-A-for-expanding}
\end{align}

To justify \eqref{equ:random-theta-A-for-expanding}, we first sample $\theta$ and condition on the event that half of $\theta_i$ satisfies that $|\theta_i| \ge \frac{1}{5\sqrt{d}} \|\theta\|_2$, which happens with probability $1-\exp(-\Omega(d))$. Then, by Lemma~\ref{lemma:rad-tail-lb}, for any small $c_{dev} > 0$, there is a suitable $c_{prob} = O({c_{dev}}^2)$ such that
$$
\Pr_{s\sim \{\pm 1\}^d} \left[ \langle s, \theta \rangle > \frac{c_{dev} \sqrt{\log N}}{\sqrt{d}} \|\theta\|_1 \right] \approx \frac{1}{N^{c_{prob}}}.
$$

It follows that with probability $1-\exp(-N^{\Omega(1)})$ over a random matrix $A\sim \{\pm 1\}^{d\times N}$, there are $\frac{1}{2}\cdot \frac{N}{N^{c_{prob}}} = \frac{1}{2} N^{1-{c_{prob}}}$ column vectors $A_{*,j}$ with $\langle \theta, A_{*,j}\rangle > \frac{c_{dev} \sqrt{\log N}}{\sqrt{d}} \|\theta\|_1 \ge \frac{c_{dev} \log N}{4}$. For such matrices $A$, if we were to sample a column uniformly at random, with probability $\frac{1}{2 N^{c_{prob}}}$ we would get a column $A_{*,j}$ such that $\langle A_{*,j}, \theta\rangle \ge \frac{c_{dev} \log N}{4}$. Next, we make use of the following lemma.

\begin{lemma}\label{lemma:exponentiation-to-score}
Let $X$ be a random variable with density $p:\mathbb{R}\to \mathbb{R}_{\ge 0}$. Suppose that $\Pr[X \ge \eta] \ge \delta$. Define a random variable $Y$ with density proportional to $q(y)\coloneqq p(y)\cdot \exp(y)$. Then, we have
$$
\E[Y]\ge \eta - 2\log(1/\delta).
$$
\end{lemma}
We defer the proof of Lemma~\ref{lemma:exponentiation-to-score} to the end of the subsection. Assuming its truth and applying it to the random variable $\langle A_{*,j}, \theta\rangle$ (where the randomness is over $j\sim [N]$), we get that for the exponential tilt $D_{\theta}(A)$, it holds that
$$
\E_{v\sim D_{\theta}}[  \langle v, \theta \rangle ] \ge \frac{c_{dev}}{4} \log N - 2c_{prob}\log (N) \ge \Omega(\log N).
$$
Here, we choose $c_{dev}\in (0, 1)$ to be sufficiently small so that we have $c_{prob} = \Theta({c_{dev}}^2) \ll c_{dev}$.

To wrap up the proof of Proposition~\ref{prop:random-is-expanding}, note that we have shown \eqref{equ:random-theta-A-for-expanding} (even with a stronger probability bound of $\exp(-\Omega(d)) + \exp(-N^{\Omega(1)})$ on the right hand side). Next, by Markov's inequality, this means that with probability $1-o(1)$ over the sampling of $A$, we have
$$
\Pr_{\theta} \left[\E_{v\sim D_{\theta}(A)}[\langle v, \theta\rangle ] \le \Omega(\log N)
 \right] < \frac{1}{d^{5}}.
$$
(If not, there would be at least $\Omega(1)\cdot \frac{1}{d^5}$ fraction of $\theta$-$A$ pairs that are ``bad'', which would contradict to \eqref{equ:random-theta-A-for-expanding}). This completes the proof of the proposition.
\end{proof}

We fill in the last gap of the proof by proving Lemma~\ref{lemma:exponentiation-to-score} below.

\begin{proof}[Proof of Lemma~\ref{lemma:exponentiation-to-score}]
Without loss of generality, we consider the case that $X\le \eta$ with probability one. Namely, it suffices to prove the lemma for the random variable $\min(X, \eta)$. We also observe that $Y$ has an exponentially decaying tail bound. Namely for every $t < \eta$, it holds that
$$
\Pr[Y \le t] \le \frac{\Pr[X\le t]}{\Pr[X \ge \eta]} \cdot \frac{\exp(t)}{\exp(\eta)} \le \frac{\exp(t - \eta)}{\delta}.
$$
Given the observation, we simply use integration by parts on $Y$ to obtain
$$
\begin{aligned}
\E[Y]
&= \eta - \int_{-\infty}^\eta \Pr[Y \le t] dt  \\
&\ge \eta - 1.9\log(1/\delta) - \int_{-\infty}^{\eta - 1.9\log(1/\delta)} \Pr[Y \le t] dt \\
&\ge \eta - 1.9\log(1/\delta) - \int_{-\infty}^{\eta-1.9\log(1/\delta)} \frac{\exp(t-\eta)}{\delta} dt \\
&\ge \eta - 1.9 \log(1/\delta) - \frac{\exp(-1.9\log(1/\delta))}{\delta} \\
&\ge \eta - 2 \log(1/\delta).
\end{aligned}
$$
This completes the proof.
\end{proof}

\subsection{Upper Bounding the Score}

\paragraph*{Regular vector family.} In this subsection, we define and study the following ``regularity'' condition on a collection of Boolean vectors.

\begin{definition}\label{def:regular-vectors}
    Let $d\in \mathbb{N}$. A collection of $d$-dimensional Boolean vectors $S = \{s^1,\dots, s^N\}$ is $r$-regular, if the following is true:
    with probability $1-\frac{1}{d^5}$ over $\theta \sim B_2(0,r)$, the exponential tilt of $S$ satisfies that
    $$
    \lambda_{max}(\COV(D_{\theta}(S))) < O\left( 1 \right).
    $$
    A matrix $A\in \{\pm 1\}^{d\times N}$ is $r$-regular, if the ensemble of column vectors of $A$ is $r$-regular.
\end{definition}

Looking ahead, a regular family of vectors gives a query matrix $A$ for which we can prove the score upper bound via our framework (c.f.~Section~\ref{sec:fingerprinting-overview}). We prove the following proposition, stating that a random query matrix is regular with desired parameters.

\begin{proposition}\label{prop:random-is-regular}
For all sufficiently small $c > 0$, the following is true. Suppose $N\gg d$ and $\log(N) \le o(d)$. Let $A\sim \{\pm 1\}^{d\times N}$ be chosen uniformly at random. With probability $1-o(1)$, the matrix $A$ is $(c\sqrt{\log N})$-regular.
\end{proposition}

\begin{proof}
Let $r = c\sqrt{\log N}$. We claim that
\begin{align}
\Pr_{\substack{\theta\sim B_2(0,r) \\ A\sim \{\pm 1\}^{d\times N}}} \left[ \lambda_{max}(\COV(D_\theta(A))) > 2 \right] \le \frac{1}{d^{10}}. \label{equ:random-theta-A-for-regular}
\end{align}
To prove \eqref{equ:random-theta-A-for-regular}, we first sample and condition on $\theta$. We advise readers to keep in mind that $\|\theta\|_2\le c\sqrt{\log N}$ where $c$ is sufficiently small.

Let $B = \{\pm 1\}^d$ be the Boolean cube. It is easy to see that $\COV(D_{\theta}(B))$ is a diagonal matrix with $\lambda_{max}(\COV(D_\theta(B))) \le 1$. Our proof strategy is to show that $\COV(D_{\theta}(A))$ is extremely close to $\COV(D_{\theta}(B))$. This happens because $A$ can be understood as sampling $N$ random vectors from $B$ with replacement. Since $N$ is extremely large, concentration inequalities apply, and we reach the desired conclusion.

We give the formal details here. First, define
$$
B' = B \setminus \{ x\in \{\pm 1\}^{d} : \langle x, \theta\rangle > 3 \sqrt{\log N} \|\theta\|_2 \}
$$
to be a pruning of $B$ where we remove all vectors that have too large an inner product with $\theta$. We claim that $D_{\theta}(B')$ is close to $D_{\theta}(B)$ in total variation distance. Indeed, we have
$$
\begin{aligned}
d_{\mathrm{TV}}(D_{\theta}(B'),D_{\theta}(B))
&= \Pr_{v\sim D_{\theta}(B)}[v\notin B'] \\
&\le \int_{3\sqrt{\log N}}^{+\infty} \Pr_{x\sim \{\pm 1\}^d}[\langle x,\theta\rangle > t \|\theta\|_2 ] \cdot \exp(t \|\theta\|_2) ~ dt \\
&\le \int_{3\sqrt{\log N}}^{\infty} \exp(-\frac{t^2}{2} + t \|\theta\|_2) ~ dt \\
&\le \frac{1}{N^2}.
\end{aligned}
$$

It then follows that $\lambda_{max}(\COV(D_{\theta}(B'))) < 1+o(1)$. Furthermore, when we sample the column vectors of $A$, with probability $1 - \frac{1}{N}$, all the vectors will be in the set $B'$. We condition on this event.

We recall a basic fact on concentration: for any bounded function $f: B' \to [-N^{1/5}, N^{1/5}]$, with probability $1-\exp(-N^{\Omega(1)})$ over $s^1,\dots, s^N\sim B'$, we have:
$$
\sum_{i \sim [N]} f(s^i) \in N \cdot \E_{v\sim B'}[ f(v) ] \pm O\left( N^{3/4} \right).
$$
We apply this fact to the following functions:
\begin{itemize}
    \item $Z(s^1,\dots, s^N) \coloneqq \sum_{i=1}^{N} \exp(\langle s^i, \theta \rangle)$.
    \item $M_u(s^1,\dots, s^N) \coloneqq \sum_{i=1}^N \exp(\langle s^i, \theta \rangle) \cdot s^i_u$ for every $u\in [d]$.
    \item $C_{u,v}(s^1,\dots, s^N) \coloneqq \sum_{i=1}^{N} \exp(\langle s^i, \theta\rangle) \cdot s^i_{u}\cdot s^i_v$, for every $u,v\in [d]$.
\end{itemize}
Note that all these functions are bounded by $N^{1/5}$, because we have that $\langle v, \theta\rangle < 3\sqrt{\log N} \|\theta\|_2 < \frac{1}{5} \log(N)$ for all $v\in B'$. Hence, we get that with high probability over $\{s^1,\dots, s^N\}$, each of these functions is close to the following functions:
\begin{itemize}
    \item $Z^* = N\cdot \E_{s\sim B'}[\exp(\langle s, \theta\rangle)]$.
    \item $M^*_u = N\cdot \E_{s\sim B'}[\exp(\langle s,\theta\rangle) \cdot s_u]$ for every $u\in [d]$.
    \item $C^*_{u,v} = N\cdot E_{s\sim B'}[\exp(\langle s, \theta\rangle) \cdot s_u\cdot s_v]$ for every $u,v\in [d]$.
\end{itemize}

Denote $A = [s^1,\dots, s^N]$. Note that the covariance matrix of $A$ is
$$
\COV(D_{\theta}(A)) = \frac{C(s)}{Z(s)} - \frac{M(s) M(s)^\top}{Z(s)^2}.
$$
Observe that $\E_{s\sim B'}[\exp(\langle s,\theta\rangle)] \ge \Omega(1)$. Then, it is straightforward to see that
$$
\begin{aligned}
\COV(D_{\theta}(A)) - \COV(D_{\theta}(B'))
&=
\left( \frac{C(s)}{Z(s)} - \frac{M(s) M(s)^\top}{Z(s)^2} \right) - \left( \frac{C^*}{Z^*} - \frac{M^*(M^*)^\top}{(Z^*)^2}  \right) \\
&\le\left( C(s)\left( \frac{1}{Z(s)} - \frac{1}{Z^*}\right) - M(s)M(s)^\top \left(\frac{1}{Z(s)^2} - \frac{1}{(Z^*)^2} \right) \right) - \\
&~~~~ - \left( \frac{C(s) - C^*}{Z^*} - \frac{M(s) M(s)^\top - M^*(M^*)^\top}{(Z^*)^2}  \right)  \\
&\le \frac{O(N^{3/4})\cdot J}{Z^*} \\
&\le O(N^{-1/4}) \cdot J.
\end{aligned}
$$
This shows that the distance between $\COV(D_{\theta}(A))$ and $\COV(D_{\theta}(B'))$ is negligible with high probability, implying that $\lambda_{max}(\COV(D_{\theta}(A)))<1+o(1)$.

Now we have shown that \eqref{equ:random-theta-A-for-regular} is true. By Markov's inequality, this implies that with probability $1-o(1)$ over $A$, we have
$$
\Pr_{\theta}\left[ \lambda_{max}(\COV(D_{\theta}(A))) > 2 \right] < \frac{1}{d^5}.
$$
(Again, if not, there would be at least $\Omega(1)\cdot \frac{1}{d^5}$ fraction of $\theta$-$A$ pairs that are bad, contradicting to \eqref{equ:random-theta-A-for-regular}). This completes the proof of the proposition.
\end{proof}

\subsection{Proof of Theorem~\ref{theo:random-query-lb}}

We are ready to conclude the proof of Theorem~\ref{theo:random-query-lb}. We condition on a matrix $A\in \{\pm 1\}^{d\times N}$ that is $(c\sqrt{\log(N)},\Omega(\log(N)))$-expanding and $(c\sqrt{\log N})$-regular for some absolute constant $c>0$. There are all but $o(1)$-fraction of such matrices by Propositions~\ref{prop:random-is-expanding} and \ref{prop:random-is-regular}. Given $A$, we use $a^i, i\in [N]$ to denote its column vectors.

Let $\alpha < \frac{c'\sqrt{\log N}}{\sqrt{d}}$ be a desired accuracy parameter, where $c' < c$ is sufficiently small. Let $n = \frac{c'\sqrt{d}}{\eps \alpha}$. Fix $\delta < \frac{1}{d}$. Suppose $\calA:[N]^n\to \mathbb{R}^{d}$ is an $(\eps,\delta)$-DP algorithm with $\ell_2^2$ error of at most $\alpha^2 d$. We derive a contradiction via the correlation analysis.

To set up the fingerprinting argument, choose $V = B_2(0,c\sqrt{\log(N)})$ and let $S$ be its surface. For every $\theta\in V$, define the exponential tilt $D_\theta$ as usual. Namely, $\Pr[D_\theta = a^i] \propto \exp(\langle \theta, e^i\rangle)$. Write $\mu_\theta\coloneqq \E[D_\theta]$ as the mean of $D_\theta$.

\paragraph*{Score lower bound.} Since the mean-squared error of $\mathcal{A}$ is $\alpha^2 d$, by Jensen's inequality, we have that $\| \E[\calA(x)] - \E_{j\sim [n]}[x^j] \|_2\le \alpha \sqrt{d}$ for every data set $x=(x^1,\dots, x^n)$. In light of this observation, we apply Proposition~\ref{prop:differentiate-to-score} and the divergence theorem to obtain
$$
\begin{aligned}
&~~~~ \E_{\theta\sim V} \E_{x\sim D_\theta^n,\calA(x)} \left[ \sum_{j=1}^{n} \langle x^j - \mu_\theta, \calA(x) \rangle \right] \\
&= \frac{\Area(S)}{\Vol(V)} \E_{\theta\sim S} \left\langle \E_{x\sim D_\theta^n,\calA(x)}[\calA(x)], \frac{\theta}{\|\theta\|_2} \right\rangle \\
&\ge \frac{d}{c\sqrt{\log N}} \left(\E_{\theta\sim S} \left\langle \E_{x\sim D_\theta^n}\E_{j\sim [n]}[x^j], \frac{\theta}{\|\theta\|_2} \right\rangle - \alpha \sqrt{d} \right) \\
&= \frac{d}{c\sqrt{\log N}} \left(\E_{\theta\sim S} \left\langle \E_{x\sim D_\theta}[x], \frac{\theta}{\|\theta\|_2} \right\rangle - \alpha \sqrt{d} \right) \\
&\ge \frac{d}{c\sqrt{\log N}} \left( (1-d^{-5})\cdot \Omega_c(\frac{\log N}{\sqrt{\log N}}) - O(d^{-5} \cdot d) - \alpha \sqrt{d} \right) & \text{(expanding of $A$)} \\
&\ge \frac{d}{c\sqrt{\log N}} \left( \Omega_c(\sqrt{\log N}) - c'\sqrt{\log N} \right) \\
&\ge \Omega_{c}(d).
\end{aligned}
$$
To achieve a tight dependence on $\alpha$, we define the score slightly differently. In particular, we observe that
$$
\E_{\theta\sim V} \E_{x\sim D_\theta^n,\calA(x)} \left[ \sum_{j=1}^{n} \langle x^j - \mu_\theta, \mu_\theta \rangle \right] = 0.
$$
Consequently,
\begin{align}
\E_{\theta\sim V} \E_{x\sim D_\theta^n,\calA(x)} \left[ \sum_{j=1}^{n} \langle x^j - \mu_\theta, \calA(x) - \mu_\theta \rangle \right] \ge  \Omega_{c}(d). \label{equ:random-query-score-lb}
\end{align}
We define the score of $x'\in [N]$ w.r.t.~$\calA(x)$ as $\score(x';\calA(x)) = \langle x' - \mu_\theta, \calA(x) - \mu_\theta \rangle$. As we will soon see, shifting $\calA(x)$ by $\mu_\theta$ allows for reducing the ``variance'' of scores, making it feasible to prove a tight upper bound on the score.

\paragraph*{Score upper bound.} To establish the score upper bound, we will use Proposition~\ref{prop:privacy-imply-small-score-difference}. The first step is to understand the variance of $\score(x',\calA(x))$ where the randomness is over $x\sim D_\theta^n, \calA(x)$ and $x'\sim D_\theta$. Since $\score(x',\calA(x))$ has zero mean, we just calculate its second moment. We proceed as
$$
\E[\score(x',\calA(x))^2] = \E\left[ (\calA(x) - \mu_\theta)^\top \COV(D_\theta) (\calA(x) - \mu_\theta)\right] \le \lambda_{max}(\COV(D_\theta)) \E[\|\calA(x) - \mu_\theta\|_2^2].
$$
By triangle inequality, we have
$$
\E[\|\calA(x) - \mu_\theta\|_2^2] \le 2\E_{x\sim D_\theta^n,\calA(x)}\left[ \|\calA(x) - \E_j[x^j]\|_2^2 \right] +2 \E_{x\sim D_\theta^n}\left[ \|\mu_\theta - \E_j[x^j]\|_2^2 \right].
$$
The first term is bounded by $\alpha^2 d$ since $\calA$ has a low mean-squared error. The second term is also bounded by $O(\alpha^2 d)$, \emph{provided that $n\ge \frac{1}{\alpha^2}$}. We can assume this is the case without loss of generality (i.e., $\frac{1}{\alpha} < \sqrt{d}$). This is because we can appeal to a folklore reduction, saying that any algorithm with $n$ samples and error $\alpha$ implies an algorithm with $\frac{n}{k}$ samples and error $k\alpha$. We prove this fact at the end of the section for completeness. (The right lower bound in the regime where $\alpha < \frac{1}{\sqrt{d}}$ can also be proven using hereditary discrepancy approaches.)

Back to our discussion, we have shown that
$$
\E[\score(x',\calA(x))^2] \le O\left( \alpha^2 d \cdot \lambda_{max}(\COV(D_\theta))\right).
$$
Now, we utilize the regular property of $A$, which tells us that $\lambda_{max}(\COV(D_\theta))\le O(1)$ with probability $1-d^{-5}$ over $\theta\sim V$. For every such $\theta$, we may use Proposition~\ref{prop:privacy-imply-small-score-difference} to conclude that
$$
\E[\score(x^j,\calA(x))] \le O(\alpha \eps \sqrt{d} + \delta d) \le O(\alpha \eps \sqrt{d})
$$
for every $j\in [n]$. 

For those $\theta$ such that $\lambda_{max}(\COV(D_\theta))$ is too large, we use the naive bound of $\E[\score(x^j,\calA(x))]\le O(d)$. Combining both cases, we conclude that
\begin{align}
\E_{\theta\sim V}\E[\score(x^j,\calA(x))] \le O(\alpha \eps \sqrt{d}). \label{equ:random-query-score-ub}
\end{align}
Unlike proofs presented in Sections~\ref{sec:fingerprinting-meets-geometry} to \ref{sec:lb-ada}, here we are only able to establish the score upper bound on average over $\theta$. Nevertheless, this is sufficient to carry out the proof.

Finally, we compare \eqref{equ:random-query-score-lb} with \eqref{equ:random-query-score-ub}. The upper and lower bounds on the score imply the desired lower bound $n\ge \Omega\left( \frac{d}{\alpha \eps \sqrt{d}} \right) \ge \Omega\left( \frac{\sqrt{d}}{\alpha \eps} \right)$, concluding the proof.

\paragraph*{Trading off accuracy for smaller sample size.} We present the promised reduction that allows one to reduce the sample size at the price of a higher error. Suppose $\mathcal{A}:[N]^n\to \mathbb{R}^{d}$ is $\alpha$-accurate and $(\eps,\delta)$-DP. We design an algorithm $\mathcal{B}:[N]^{n/k}\to \mathbb{R}^{d}$ that is $\alpha k$-accurate and retains $(\eps,\delta)$-DP. Denote $m = n/k$. Say $(x^1,\dots, x^{m})$ is the input to $\mathcal{B}$. Take $z$ to be arbitrary. Make $n-m$ copies of $z$ and feed $\mathcal{A}$ with the input $(x^1,\dots, x^{m}, z,\dots, z)$. Let $q\in \mathbb{R}^{d}$ be the output of $\mathcal{A}$. We let $\mathcal{B}$ output $\frac{n}{m}\left(q - \frac{n - m}{n}\cdot z \right)$. $\mathcal{B}$ is clearly private. Regarding accuracy, we have
$$
\left\| \frac{n}{m}(q-\frac{n - m}{n}\cdot z) - \frac{1}{m}\sum_{j=1}^m x^j \right\|_2 = \left\| \frac{n}{m} \left( q - \frac{1}{n} \left( \sum_{j=1}^{m} x^j + (n - m)\cdot z\right) \right) \right\|_2.
$$
Note that the right hand side measures the error of $\mathcal{A}$ on $(x^1,\dots, x^m,z,\dots, z)$ and multiply it with $\frac{n}{m} = k$. Hence, assuming $\mathcal{A}$ has error $\alpha$, it follows that $\mathcal{B}$ has error $k\alpha$.

\section{Releasing Random Linear Queries via Sparse Histogram} \label{sec:random-ub}

In this section, we present our algorithm for answering random linear queries with near-optimal sample complexity in the low-accuracy regime. The main result covered in this section is Theorem~\ref{theo:query-release-via-histogram}, to appear in Section~\ref{sec:random-ub-finish}.

\subsection{The Structure of Random Queries}

We need the following structural result concerning random matrices. Roughly, it says given a random query matrix $A\in \{\pm 1\}^{d\times N}$, any small number of column vectors behave like mutually orthogonal vectors.

\begin{lemma}\label{lemma:random-query-orthogonal}
    Let $N,d\ge 0$ be two integers. Consider a random matrix $A:\{\pm 1\}^{d\times N}$ where each entry of $A$ is independently set to $\pm 1$ with equal probability.

    Then, with probability $1-o(1)$ over $A$, the following is simultaneously true for all $k\in \left[1, c\cdot \frac{d}{\log(N)}\right]$: for all possible subset of $k$ distinct columns $i_1,\dots, i_k \in [N]$, the column vectors $A_{*, i_1},\dots, A_{*, i_{k}}$ satisfy $\left\|\sum_{j} A_{*,i_j} \right\|_2 \le \sqrt{2kd}$.
\end{lemma}

Before starting the proof, we introduce one useful concentration inequality about quadratic forms of sub-gaussian random variables, known as the Hanson-Wright inequality.

\begin{lemma}[\cite{Hanson-Wright-ineq}]\label{lemma:Hanson-Wright}
There is a universal $c > 0$ for which the following is true. Let $X = (X_1,\dots, X_n)$ be a vector of independent random variables that are all zero-mean and $K$-subgaussian. Let $M$ be an $n\times n$ matrix. Then, for every $t\ge 0$, it holds that
$$
\Pr[|X^TMX - \E[X^T M X] | > t] \le 2\exp\left(-c \min\left(\frac{t^2}{K^4\|M\|_F^2}, \frac{t}{K^2 \|M\|_{op}} \right) \right).
$$
\end{lemma}

\begin{proof}[Proof of Lemma~\ref{lemma:random-query-orthogonal}]
    Fix one $k$. Let $i_1,\dots, i_k$ be a list of fixed indices. Observe that $\sum_{j=1}^{k} A_{*,j}$ is a vector of $d$ independent entries where each entry is $\sqrt{k}$-subgaussian (this follows because the sum of $k$ Bernoulli random variables is $\sqrt{k}$-subgaussian). As such, we use Lemma~\ref{lemma:Hanson-Wright} by taking $M$ there to be $I$, and derive that
    $$
    \Pr\left[ \left\| \sum_{j} A_{*,j} \right\|_2^2 - kd > t\right] < 2\exp\left( -c \min\left( \frac{t^2}{k^2 d}, \frac{t}{k} \right) \right).
    $$
    By taking $t = kd$, we see that
    $$
    \Pr\left[ \left\| \sum_{j} A_{*,j} \right\|_2 > \sqrt{2 kd} \right]  =\Pr\left[ \left\| \sum_{j} A_{*,j} \right\|_2^2 > 2 kd\right] < \exp(-\Omega(d)).
    $$
    We can then union-bound over all possible choices of $i_1,\dots, i_k$. Note that there are $N^{k} = 2^{\frac{c d}{\log(N)}\cdot \log(N)} < 2^{c\cdot d}$ of them. Lastly, we union-bound over all possible $k$'s (there are at most $\sqrt{d}$ of them).
\end{proof}

\subsection{Private Sparse Histogram}

We need the following well-known result from the differential privacy literature.

\begin{lemma}\label{lemma:private-sparse-histogram}
Let $\mathcal{X}$ be a (possibly unbounded) universe and $n\ge 1$ be a finite integer. There is an algorithm which, on input a non-negative vector $x\in \mathbb{R}_{\ge 0}^{\mathcal{X}}$ with $\| x\|_1 = n$, with probability one returns a vector $\hat{x}$ with $\| \hat{x} \|_1 = n$ such that
$$
\| x - \hat{x} \|_{\infty} < O\left( \frac{1}{\eps} \log(1/\delta) \right).
$$
Furthermore, the algorithm is $(\eps,\delta)$-DP with respect to any adjacent inputs $x,x'$ with $\| x - x'\|_1 < 1$.
\end{lemma}

Most previous works using or designing sparse histogram algorithms are concerned with the typical definition of add-remove or change-one privacy. For our purpose, we need a version of the sparse histogram algorithm to handle $\ell_1$-adjacent data sets. For completeness, we include a proof of Lemma~\ref{lemma:private-sparse-histogram} below.

\begin{proof} We start by introducing the truncated Laplace mechanism and giving a ``fine-grained'' privacy analysis for it, taking the distance between two near inputs into account.

    \medskip\noindent\textbf{Truncated Laplace mechanism.} For $v,\eps > 0$, let $\TruncLap_{v}(\frac{1}{\eps})$ be a truncated Laplace random variable, defined as the following. To sample from $\TruncLap_{v}(\frac{1}{\eps})$, one \emph{repeatedly\footnote{There is another popular way of defining truncated Laplace noise by drawing a single sample from $\mathrm{Lap}(1/\eps)$ and truncating it into the range $[-v, v]$. However, for our algorithm and its analysis, the rejection sampling version is needed.}} draws $A\sim \mathrm{Lap}(\frac{1}{\eps})$, and returns the first $A$ such that $A\in [-v, v]$.

    Now, let $\eps,\delta$ be given as in Lemma~\ref{lemma:private-sparse-histogram}. We choose $v = \frac{5\log(1/\delta)}{\eps}$ and make the following claim: if two real numbers $X, Y$ are such that $|X - Y| = \Delta\le 1$, then $X + \TruncLap_{v}(\frac{1}{\eps})$ and  $Y + \TruncLap_{v}(\frac{1}{\eps})$ are $(\Delta\eps,\Delta \delta)$-indistinguishable. To briefly justify this, consider the output distribution of, e.g., $X + \TruncLap_{v}(\frac{1}{\eps})$. We consider getting a sample from $\TruncLap_{v}(\frac{1}{\eps})$ with value inside $[-v,-v+\Delta]\cup [v-\Delta, v]$ as a privacy failure. The probability of failure is then easily shown to be at most $\Delta \delta$. If the failure event does not happen, the output $X+\TruncLap_v(1/\eps)$ is equally likely (up to a multiplicative factor of $e^{\Delta \eps}$) to occur as well when we sample from $Y + \TruncLap_{v}(\frac{1}{\eps})$ instead.

    \medskip\noindent\textbf{Algorithm.} Let us design the algorithm. For each $u\in \mathcal{X}$, we tentatively set $\tilde{x}_u = \max(0, x_u + \TruncLap_v(\frac{1}{\eps}))$.
    This gives us a vector $\tilde{x}$ which satisfies that
    $$
    \|\tilde{x} - x\|_{\infty} \le v.
    $$
    Currently, $\tilde{x}$ does not necessarily satisfy that $\|\tilde{x}\|_1 = n$. Consider the set $H = \{ y\in \mathbb{R}_{\ge 0}^{\mathcal{X}}: \|y\|_1 = n\}$. We project $\tilde{x}$ onto $H$ by minimizing $\ell_{\infty}$ movement. Denote the resulting point by $\hat{x}$. Since the vector $x$ is in $H$, and $x$ is close to $\tilde{x}$ in $\ell_{\infty}$ distance, an application of triangle inequality shows that
    $$
    \| \hat{x} - x \|_{\infty} \le 2 \| \tilde{x} - x \|_{\infty}  = 2v.
    $$

    \medskip\noindent\textbf{Privacy.} The above concludes the description of the algorithm as well as its utility analysis. In terms of privacy, let $x,x'$ be a pair of vectors such that $\|x-x'\|_1 \le 1$. For every element $u\in \mathcal{X}$, the value $\tilde{x}_u$ and $\tilde{x'}_u$ will be $(\eps |x_u - x'_u|,\delta | x_u - x'_u |)$-indistinguishable. Hence, using the basic composition of differential privacy, we conclude that the two vectors $\tilde{x}, \tilde{x'}$ are $(\eps,\delta)$-indistinguishable. Since $\hat{x}$ and $\hat{x'}$ are post-processing of $\tilde{x},\tilde{x'}$, this projection step does not change our privacy analysis.
\end{proof}

\subsection{The Query Releasing Algorithm} \label{sec:random-ub-finish}

We are ready to design our query-releasing algorithm for random linear queries.

\begin{theorem}\label{theo:query-release-via-histogram}
Let $N,d\ge 0$ be two integers. Consider a random matrix $A\in\{\pm 1\}^{d\times N}$ where each entry of $A$ is independently set to $\pm 1$ with equal probability.

Then, with probability $1-o(1)$ over $A$, the following is simultaneously true for all $\alpha \in \left( C\sqrt{\frac{\log N}{d}}, 1 \right)$: for every $\eps,\delta \in (0, 1)$ there is a bound $n = O\left( \frac{\log(1/\delta)}{\alpha^2 \eps} \right)$ and an $(\eps,\delta)$-DP algorithm $\calA$ such that for every input vector $x\in \mathbb{R}^{N}$ with $\|x\|_1 \ge n$, $\calA$ with probability one returns a $\hat{y}$ such that $\left\| A \frac{x}{\|x\|_1} - \hat{y} \right\|_2 < \alpha \sqrt{d}$. Furthermore, $\calA$ is $(\eps,\delta)$-DP w.r.t.~$\ell_1$-adjacent data sets.
\end{theorem}

\begin{proof}
We first state the algorithm, and then give its privacy and utility analysis.

\medskip\noindent\textbf{Algorithm.} Given input $x$, we first scale it properly to make it the case that $\|x\|_1 = n$ (namely we change $x$ to $\frac{n x}{\|x\|_1}$). Then, we use the sparse histogram algorithm of Lemma~\ref{lemma:private-sparse-histogram} on $x$ to find a vector $\hat{x}$ such that $\| x - \hat{x} \|_\infty < O(\frac{1}{\eps}\log(1/\delta))$. Finally, we output $\frac{A\hat{x}}{\|\hat{x}\|_1}$.

\medskip\noindent\textbf{Privacy.} To understand the privacy of the algorithm, note that for two adjacent $x,x'$ with $\min(\|x\|_1,\|x'\|_1)\ge n$, the scaling step increases their $\ell_1$ distance by at most $1$:
$$
\begin{aligned}
\left\| \frac{ n x}{\|x\|_1} - \frac{n x'}{\|x'\|_1} \right\|_1
& \le \left\| \frac{ n x}{\|x\|_1} - \frac{n x'}{\|x\|_1} \right\|_1 + \left\| \frac{ n x'}{\|x\|_1} - \frac{n x'}{\|x'\|_1} \right\|_1 \\
& \le \frac{n}{\|x\|_1} \|x - x'\|_1 + n \|x'\|_1 \cdot \left| \frac{1}{\|x\|_1} - \frac{1}{\|x'\|_1} \right| \\
&\le 1 + \frac{n \|x'\|_1}{\|x\|_1 \cdot \|x'\|_1} \\
&\le 2.
\end{aligned}
$$
Therefore, to ensure the final algorithm is $(\eps,\delta)$-DP, we can work with slightly smaller privacy parameters $(\eps/2,\delta/2)$ when invoking Lemma~\ref{lemma:private-sparse-histogram}. This only blows up the relevant parameters by a constant factor.

\medskip\noindent\textbf{Utility.} We now prove the utility of the algorithm. We assume the matrix $A$ is such that the conclusion of Lemma~\ref{lemma:random-query-orthogonal} holds (which happens with probability $1-o(1)$). We also assume the input $x$ has $\|x\|_1 = n$. Let $\hat{x}$ be the algorithm's output. From Lemma~\ref{lemma:private-sparse-histogram} we see that
$$
\| \hat{x} - x \|_{1} \le \|\hat{x}\|_1 + \|x\|_1 \le 2n
$$
and
$$
\| \hat{x} - x\|_{\infty} \le v
$$
for some $v\le O(\frac{\log(1/\delta)}{\eps})$.

For every subset $J\subseteq [N]$ of $|J|\le\frac{2n}{v}$ columns, let $e_J\in \mathbb{R}^{N}$ be a vector which takes value $1$ on coordinates from $J$ and equals zero elsewhere. It is easy to see that $\hat{x} - x$ can be written as a convex combination of $\{\pm v\cdot e_J : J\subseteq [N], |J| \le 2n/v\}$. Here, we assume $\alpha > C\sqrt{\frac{\log N}{d}}$ for a large enough $C$, so that by setting $n = \frac{C' \log(1/\delta)}{\eps \alpha^2}$ for an appropriate $C'$, the cardinality of $J$ (i.e.,~$\frac{2n}{v} = \frac{C'}{C^2}\frac{d}{\log N}$) falls under the regime of Lemma~\ref{lemma:random-query-orthogonal}. Then, for every $J$, by Lemma~\ref{lemma:random-query-orthogonal} we have
$$
\| A e_J \|_2 \le 2\sqrt{d  |J|} \le 4 \sqrt{\frac{d n}{v}}.
$$
Finally, by an averaging argument, we see that
$$
\begin{aligned}
\left\| \frac{Ax}{\|x\|_1} - \frac{A\hat{x}}{\|\hat{x}\|_1} \right\|_2
&= \frac{1}{n} \| A(x - \hat{x}) \|_2 \\
&\le \frac{1}{n} \sup_{J} \{ \| v\cdot A e_J \|_2 \} \\
&\le \frac{1}{n} \cdot v \cdot \frac{4\sqrt{dn}}{\sqrt{v}} \\
& \le 4 \sqrt{d\cdot \frac{v}{n}} \\
& \le \alpha \sqrt{d}.
\end{aligned}
$$
The last inequality is valid so long as we choose $C'$ to be large enough, so that with $n = \frac{C' \log(1/\delta)}{\eps \alpha^2}$ and $v = O\left( \frac{\log(1/\delta)}{\eps} \right)$ (this big-Oh hides a constant independent of $\alpha, n, d$ and $N$), we still have $\frac{v}{n} < \frac{\alpha^2}{16}$. This completes the proof.
\end{proof}

%\section{Conclusions}
\addcontentsline{toc}{section}{References}
\bibliographystyle{alpha}
\bibliography{submission/references}

\end{document}